\documentclass[10pt,journal]{IEEEtran}
\usepackage{amsmath, amsthm, amssymb, amsfonts}
\usepackage{textcomp}
\usepackage{graphicx}
\usepackage{cite}

\usepackage{ulem}
\usepackage{algorithmicx}
\usepackage{algorithm}
\usepackage{array}
\usepackage[caption=false,labelfont=sf,textfont=sf]{subfig}
\usepackage[implicit=true,hidelinks]{hyperref}
\usepackage{stfloats}
\usepackage{url}
\usepackage{verbatim}
\usepackage{multirow,booktabs}
\usepackage{xcolor}
\usepackage{makecell}
\usepackage{ragged2e}
\usepackage{threeparttable}
\usepackage{pgfplots}
\usepackage{tikz}
\usepackage{colortbl}
\usepackage{color}
\usepackage{orcidlink}
\usetikzlibrary{arrows}
\usetikzlibrary{math}
\usetikzlibrary{positioning}

\usetikzlibrary{quotes,angles}
\usetikzlibrary{calc}
\usetikzlibrary{decorations.pathreplacing}
\usetikzlibrary{shapes.arrows}
\usepackage{pgfplotstable}
\pgfplotsset{compat=1.7}
\usepackage [
lambda, 
advantage , operators , sets,
adversary , landau , probability , notions ,
logic, ff, mm,
primitives , events , complexity , oracles , asymptotics , keys
]{cryptocode}

\definecolor{seagreen}{rgb}{0.18, 0.55, 0.34}
\definecolor{darkorange}{rgb}{1.0, 0.55, 0.0}
\definecolor{deeppink}{rgb}{1.0, 0.08, 0.58}
\definecolor{blue-violet}{rgb}{0.54, 0.17, 0.89}
\definecolor{brandeisblue}{rgb}{0.0, 0.44, 1.0}
\definecolor{carminered}{rgb}{1.0, 0.0, 0.22}
\newcommand{\ek}{\ensuremath{\mathsf{ek}}}

\newcommand{\C}{\ensuremath{\mathsf{C}}}

\newcommand{\M}{\ensuremath{\mathsf{M}}}
\newcommand{\calS}{\ensuremath{\mathcal{S}}}

\newtheorem{definition}{Definition}
\newtheorem{theorem}{Theorem}

\definecolor{color1}{RGB}{0,114,178} 
  \definecolor{color2}{RGB}{213,94,0}   
  \definecolor{color3}{RGB}{0,158,115}  
  
\def\BibTeX{{\rm B\kern-.05em{\sc i\kern-.025em b}\kern-.08em
    T\kern-.1667em\lower.7ex\hbox{E}\kern-.125emX}}

\begin{document}

\title{Silent Guardians: Independent and Secure Decision Tree Evaluation Without Chatter}

\author{Jinyuan Li$^{\orcidlink{https://orcid.org/0009-0007-7703-8428}}$ and 
Liang Feng Zhang$^{\orcidlink{https://orcid.org/0000-0003-3543-1524}}$
\thanks{
  Jinyuan Li and Liang Feng Zhang are with the School of Information Science and Technology, 
  ShanghaiTech University, Shanghai 201210, China. E-mail: {lijy6, zhanglf}@shanghaitech.edu.cn.
  }
}

\markboth{Journal of \LaTeX\ Class Files,~Vol.~14, No.~8, August~2021}%
{Shell \MakeLowercase{\textit{et al.}}: A Sample Article Using IEEEtran.cls for IEEE Journals}

\maketitle

\begin{abstract}
As machine learning as a service (MLaaS) gains increasing popularity, it raises two critical challenges: privacy and verifiability. 
For privacy, clients are reluctant to disclose sensitive private information to access MLaaS, while model providers must safeguard their proprietary models. 
For verifiability, clients lack reliable mechanisms to ensure that cloud servers execute model inference correctly.
Decision trees are widely adopted in MLaaS due to their popularity, interpretability, and broad applicability in domains like medicine and finance. 
In this context, outsourcing decision tree evaluation (ODTE) enables both clients and model providers to offload their sensitive data and decision tree models to the cloud securely. 
However, existing ODTE schemes often fail to address both privacy and verifiability simultaneously. 
To bridge this gap, we propose $\sf PVODTE$, 
a novel two-server private and verifiable  ODTE protocol that leverages homomorphic secret sharing and a MAC-based verification mechanism. 
$\sf PVODTE$ eliminates the need for server-to-server communication, enabling independent computation by each cloud server.  
This ``non-interactive'' setting 
addresses the latency and synchronization bottlenecks of prior arts, 
making it uniquely suitable for wide-area network (WAN) deployments. 
To our knowledge, $\sf PVODTE$ 
is the first two-server ODTE protocol that eliminates server-to-server communication. 
Furthermore, $\sf PVODTE$ achieves security against \emph{malicious} servers, 
where servers cannot learn anything about the client's input or the providers' decision tree models, 
and servers cannot alter the inference result without being detected. 
\end{abstract}

\section{Introduction}
\IEEEPARstart{M}{achine} Learning as a Service (MLaaS) enables access to advanced machine learning models without requiring extensive in-house resources or expertise. 
Driven by cloud computing, MLaaS has evolved into a paradigm~\cite{zheng2022} where a model provider (e.g., a hospital) deploys a pre-trained model (e.g., a diagnostic decision tree) to the cloud, offering inference services to clients (e.g., patients) who submit sensitive input data (e.g., health records). Outsourcing MLaaS to the cloud delivers well-documented benefits for both parties: high scalability, ubiquitous access, and low economical cost.

However, this paradigm raises two fundamental security
challenges~\cite{riasi2024}: 
(i) Privacy: Model providers (e.g., hospitals)
  hesitate to outsource proprietary diagnostic models to cloud
  service providers due to risks of intellectual property
exposure; clients (e.g., patients) fear privacy violations when
transmitting plaintext sensitive data (e.g., health records). 
(ii)
Verifiability: Untrusted cloud servers may return incorrect
inference results (e.g., manipulated medical diagnoses), but
clients lack mechanisms to verify result correctness.

To address both privacy and verifiability in MLaaS, privacy-preserving and verifiable model inference (PVMI) 
has emerged as a key research direction. 
PVMI protocols guarantee two core properties: 
(i) cloud servers cannot access the provider's proprietary model or client's sensitive input in plaintext; 
(ii) malicious server behavior (e.g., result manipulation) is detectable.
PVMI solutions are broadly categorized into single-server and multi-server paradigms based on the number of cloud servers involved.

In the single-server setting, 
the provider's encrypted model and client's encrypted input are sent to one cloud server, which performs encrypted inference and returns an encrypted result.
Prior works~\cite{riasi2024,Weng2023,lu2018,aloufi2021} use additively homomorphic encryption (AHE)~\cite{cramer1997} or fully homomorphic encryption (FHE)~\cite{gentry2009} to outsource models like 
neural networks and decision trees. 
For verifiability, some works~\cite{riasi2024,Weng2023} integrate zero-knowledge proofs (ZKPs)~\cite{groth2016a}, 
while others~\cite{lu2018,aloufi2021} lack such mechanisms, leaving clients unable to verify the correctness of inference results. 
However, single-server designs suffer from \emph{prohibitive computational overhead} due to AHE/FHE, limiting their real-world practicality~\cite{dong2023}.

To mitigate these computational bottlenecks, multi-server architectures use secret sharing (SS)~\cite{shamir1979a,beaver1992,Cramer2005}, where 
the model and input are split into shares, and then distributed across two or more non-colluding servers~\cite{liu2023,he2024,tsuchida2020,ji2023,cheng2024,liu2019,ma2021,zheng2022,zheng2023}. 
SS-based protocols replace expensive homomorphic operations with lightweight arithmetic, improving efficiency. 
However, they introduce new challenges:
high network latency and 
excessive bandwidth consumption due to frequent server-to-server (S2S) communication (e.g., data exchange between servers) during inference, especially in wide-area networks (WANs).
Moreover, existing SS-based PVMI protocols~\cite{liu2023,he2024,tsuchida2020,ji2023,cheng2024,liu2019,ma2021,zheng2022,zheng2023} 
only satisfy privacy but lack verifiability, clients  cannot detect malicious server behavior like result manipulation.

In this work, we focus on a specific type of machine learning model: decision trees, 
which are popular, interpretable, and widely used in applications like medical diagnosis~\cite{liang2021,zhang2023}, financial distress prediction~\cite{qian2022}, and more~\cite{Francescomarino2019,zhengy2022}. 
Specifically, we target the \emph{outsourced decision tree evaluation} (ODTE) problem: securely outsourcing a decision tree model and the client input to the cloud, with inference performed on encrypted or secret-shared data 
such that the model and input remain private, and the result is  verifiable.

Prior work on ODTE prioritizes privacy (i.e., hiding the model and the client input from servers). 
According to the number of required servers, existing ODTE protocols can be classified into three categories:
(i) \emph{single-server} protocols~\cite{lu2018,aloufi2021} that minimize server requirements but suffer from prohibitive computational overhead due to the reliance on FHE;
(ii) \emph{two-server} protocols~\cite{liu2019,ma2021,zheng2022,zheng2023} that retain a relatively simple setup but  
require $\mathcal{O}(h)$, $\mathcal{O}(\secpar)$ or $\mathcal{O}(\log (\secpar))$ S2S rounds, 
where $h$ is the height of the decision tree and 
$\secpar$ is the security parameter, leading to high network latency; 
(iii) \emph{three-server} protocols~\cite{tsuchida2020,ji2023,cheng2024} that reduce S2S rounds to constant but increase deployment complexity (three cloud servers).
Unfortunately, all existing ODTE protocols suffer from two critical limitations: 
(i) \emph{Performance bottlenecks}: They suffer from either significant computational overhead (single-server designs) or high network latency (multi-server designs);
(ii) \emph{Weak security guarantees}: No protocol provides security against \emph{malicious} servers. Malicious servers may compromise the privacy of clients and model providers, or manipulate inference results.

These gaps motivate our core research question:
\emph{Can we design an ODTE protocol that simultaneously achieves three goals: (i) maintaining a practical two-server setup to avoid FHE; (ii) eliminating S2S rounds to minimize network latency; and (iii) achieving security against malicious servers to ensure privacy and verifiability?}

\begin{table}[t]
  \centering
  \renewcommand{\arraystretch}{1.3}
  \caption{
    Comparison of ODTE protocols.
  }
    \resizebox{1\columnwidth}{!}{
      \begin{threeparttable}
        
    \begin{tabular}{l||c|c|c|c}
      \hline
      { Protocol}  & {\#Servers}
    &
     \makecell[c]{ { Malicious} \\
     { security}  } &
        { \#S2S Rounds} & {Techniques}
        \\ 
         \hline
        \cite{aloufi2021}   & $1$ &  $\circ$  & $\bf 0$ & FHE     \\
        \hline
        \cite{lu2018}       & $1$ &   $\circ$  & $\bf 0$  & FHE   \\ 
        \hline \hline
        \cite{tsuchida2020} & $3$  & $\circ   $  & $25$ & RSS     \\
        \hline
        \cite{cheng2024}    & $3$  & $\circ   $  & $5$ & FSS, RSS      \\
        \hline
        \cite{ji2023}       & $3$  & $\circ   $  & $4$  & FSS    \\
        \hline \hline
        \cite{liu2019}      & $2$  & $\circ   $  & $h+3$ & SS, HE    \\
        \hline 
        \cite{ma2021}       & $2$ & $\circ $  & $2h$ &SS, OT    \\
        \hline  
        \cite{xu2023}     & $2$ & $\circ   $  & $2(h+1)$ & FHE   \\
        \hline 
        \cite{zheng2022}    & $2$ & $\circ   $  & $2 (\secpar-1)$ & SS  \\
        \hline 
        \cite{zheng2023}    & $2$  & $\circ   $  & $\log (\secpar)+5$ & SS, OT   \\
        \hline 
        Ours                & $2$  & $\bullet $  & $\bf 0$ & HSS     \\
       \hline
      \end{tabular}
      Remark: 
        A protocol that possesses a particular property in a given column is marked with a $\bullet$; 
        otherwise, it is marked with a $\circ$. 
        \#Servers denotes the number of required servers. 
        \#S2S Rounds denotes the rounds of S2S communication.
        $h$ denotes the height of the decision tree, 
        $\secpar$ denotes the security parameter.
    \end{threeparttable}
    }
    \vspace{-4mm}
\label{table:comp_theoretically}
\end{table}

\subsection{Our Contributions}

In this paper, we propose $\sf PVODTE$. To our knowledge, this is the first two-server ODTE protocol that \emph{eliminates S2S rounds} while simultaneously achieving security against \emph{malicious} servers. A detailed comparison with existing ODTE protocols is provided in TABLE \ref{table:comp_theoretically}. Furthermore, we demonstrate the practical applicability of $\sf PVODTE$ on gradient boosted decision trees and features with categorical values.

Our key tool to obtain $\sf PVODTE$ is a \emph{new} secure integer comparison algorithm designed for two non-communicating servers. 
This algorithm enables efficient integer comparison without any S2S rounds, 
while guaranteeing that neither server learns the plaintexts of the inputs or the comparison result. 
To our knowledge, this is the \emph{first} such algorithm tailored for 
two non-communicating servers.  
Compared with existing secure integer comparison algorithms, 
our algorithm retains the $\mathcal{O}(t)$ multiplication gate complexity (where $t$ is the bit length of the inputs) while eliminating the need for S2S rounds. 
Detailed performance comparisons are provided in TABLE \ref{table:sic}. 
We further extend this algorithm to support secure floating-point number comparison. 
Notably, the secure integer comparison algorithm is not only critical for decision tree evaluation  
but also useful in other applications~\cite{Damg2007,Gentry2015}.

We implement a proof-of-concept system for $\sf PVODTE$ to validate its practicality and efficiency. 
We conduct empirical evaluations using decision trees of practical sizes (consistent with real-world deployments) to benchmark our protocol's performance. 
Experimental results show that our protocol achieves a $20\times$ improvement in total communication efficiency over the state-of-the-art ODTE protocol~\cite{zheng2023}. 
Additionally, our protocol's overall computation speed is $2.7\times$ 
faster than that in~\cite{zheng2023}. 
The detailed comparison results are presented in TABLE \ref{tab:serverlatency} and TABLE \ref{tab:online communication}.
The source code is publicly available at \url{https://github.com/neuroney/PODT}.

\begin{table}[t]
  \centering
  \renewcommand{\arraystretch}{1.3}
  \caption{Comparison of Two-server Secure Integer \\ Comparison algorithms.}
    \resizebox{1\columnwidth}{!}{
    \begin{tabular}{c||c|c|c}
      \hline
      Algorithms &  \#S2S Rounds & \#Mult. & \#Comm. \\
      \hline
      \cite{zheng2022} & $2 t-2$ & $3t-5$ & $12t^2-10t$  \\
      \hline
      \cite{zheng2023} & $\log t+1$ & $3t-5$ & $12t^2-10t$  \\
      \hline
      \cite{Damagrd2006} & $44$ & $205t+188\log t$ & $410t^2+376t^2\log t$  \\
      \hline
      \cite{Nishide2007} & $15$ & $279t+5$ & $558t^2+10t$  \\
      \hline
      \cite{Reistad2007} & $8$ & $27t+36t+5$ & $54t^2+72t\log t+10t$  \\
      \hline
      \cite{Sutradhar2023} & $23$ & $10$ & $84t$  \\
      \hline
      Ours & $\bf 0$ & $4t-2$  & $\bf 0$  \\
      \hline
      \end{tabular}
      }
      \begin{justify}
        Remark: $t$ denotes the number of bits of an integer compared. \#S2S Rounds denotes the number of rounds of S2S communication. 
        \#Mult. denotes the number of multiplication gates.
        \#Comm. denotes the number of bits of communication.
      \end{justify}
\label{table:sic}
\vspace{-8mm}
\end{table}

$\sf PVODTE$ targets privacy-preserving MLaaS where model owners and clients are offline. 
The two non-colluding servers can be instantiated as organizationally independent cloud providers (e.g., AWS and Azure), making collusion unlikely in practice. 
By eliminating online S2S rounds, 
$\sf PVODTE$ excels in high-throughput WAN environments with a large number of queries, 
effectively amortizing the $\mathcal{O}(mn)$ one-time download cost ($m$ non-leaf nodes; $n$ features), though is less efficient for sporadic queries or LAN settings. 
For deployments where the one-time downloading is unacceptable such as stateless clients, an alternative instantiation over RLWE-based \textsf{HSS}~\cite{boyle2019} can eliminate it entirely, at the cost of increased server-side computation; we discuss this variant in detail in the supplementary material.

\subsection{Our Designs}
\label{sec:design}

Our $\sf PVODTE$ relies on two core building blocks: 
homomorphic secret sharing (HSS)~\cite{boyle2019,boyle2016,orlandi2021,abram2022} and message authentication codes (MACs). 
HSS is selected for enabling non-interactive private decision tree evaluation across two non-communicating servers, while MACs are adopted for generating verifiable evaluation results. 

Following~\cite{kiss2019}, 
the procedure for decision tree evaluation in ODTE consists of three main algorithms: secure feature selection ($\sf SFS$), secure integer comparison ($\sf SIC$), and secure result generation ($\sf SRG$).  
In $\sf PVODTE$, we use HSS to instantiate these three algorithms, 
enabling a non-interactive privacy-preserving evaluation process.
To achieve verifiability, we augment the $\sf SRG$ algorithm with MACs, yielding the fourth algorithm: verifiable result generation 
($\sf VRG$).

\noindent{\bf Secure feature selection.}
The goal of secure feature selection is to obliviously select, 
for each of the $m$ decision nodes, a corresponding feature from the client's $n$-dimensional feature vector $\mathbf{x}$. 
This selection follows the decision tree's feature map $\delta: [m] \to [n]$, where $\delta(j)$ indexes the feature for the $j$-th decision node. 
To represent $\delta$ in a form amenable to HSS evaluation, the model provider encodes $\delta$ as a binary matrix $\mathcal{M} \in \{0,1\}^{m \times n}$, where each row $\mathcal{M}_j$ contains exactly one non-zero entry at column $\delta(j)$, so that $x_{\delta(j)} = \mathcal{M}_j \cdot \mathbf{x}$. The model provider encrypts $\mathcal{M}$ under HSS to obtain $\C_{\mathcal{M}}$, which is published as a public parameter accessible to all clients. Each client downloads $\C_{\mathcal{M}}$ once as a one-time setup cost; for deployments where this cost is unacceptable, we describe an alternative instantiation that eliminates the download entirely in the supplementary material.
Given a feature vector $\mathbf{x}$, the client runs the $\sf SFS$ algorithm to locally compute an $m \times t$ matrix $\C_{\mathcal{M}\mathbf{x}}$, where each entry $\C_{(\mathcal{M}\mathbf{x})_{j,i}}$ is an HSS ciphertext encrypting the $i$-th bit of $x_{\delta(j)}$. This bit-decomposed representation directly serves as input to the subsequent $\sf SIC$ step. The client then sends $\C_{\mathcal{M}\mathbf{x}}$ to the two cloud servers. The semantic security of HSS guarantees that the client learns nothing about the feature map $\delta$ from $\C_{\mathcal{M}}$, while the servers obtain no information about the feature indices or values from $\C_{\mathcal{M}\mathbf{x}}$.

\noindent\textbf{Secure integer comparison.}
The secure integer comparison step compares the selected feature value $x_{\delta(j)}$ against the threshold $y_j$ (for the $j$-th decision node). 
To enable non-interactive secure comparison, 
we adopt an iterative approach 
to construct the comparison result. 
Specifically, 
let $\alpha = \sum_{i=1}^t 2^{i-1} \cdot \alpha_i$ and $\beta = \sum_{i=1}^t 2^{i-1} \cdot \beta_i$ be two $t$-bit integers, where $\alpha_i, \beta_i \in \{0,1\}$ represent their $i$-th binary bits (with $\alpha_1$ and $\beta_1$ be the least significant bits). 
The $\sf SIC$ algorithm iteratively constructs the comparison result by evaluating the difference of substrings. 
Define $c_0 = 0$, and let $c_i, 1\leq i \leq t$ represent the result of comparing the $i$-bit substrings $\alpha_i \dots \alpha_1$ and $\beta_i \dots \beta_1$ (i.e., whether $\alpha_i \dots \alpha_1 > \beta_i \dots \beta_1$). The value of $c_{i+1}$ is recursively defined as:
\begin{equation}
c_{i+1}:=\  (\alpha_{i+1}>\beta_{i+1})\vee[(\alpha_{i+1}=\beta_{i+1})\wedge c_i], 0\leq i < t, \label{eq:1}
\end{equation}
where $c_t$ yields the final comparison result $(\alpha > \beta)$. 
In $\sf SIC$ algorithm, we transform Eq.~\eqref{eq:1} into a restricted multiplication straight-line (RMS) program~\cite{cleve1990}, 
which can be efficiently evaluated by HSS.
$\sf SIC$ ensures that neither server learns any information about the compared integers ($\alpha, \beta$) or the  result. 
Thus, the privacy of $x_{\delta(j)}$, $y_j$, and their comparison result is guaranteed against the cloud servers.

\noindent{\bf Secure Result Generation.}  
After obtaining the compared result at each decision node, 
the inference result is generated using the ``path costs'' methodology from~\cite{tai2017}. 
Briefly speaking, each leaf node $L_j$ of the decision tree is assigned a path cost $pc_j$, which is determined by the comparison results of the decision nodes along the path to $L_j$. 
The correct inference result $\mathcal{T}({\bf x})$ (where $\mathcal{T}$ denotes the decision tree model) is the classification label $v_j$ of leaf node $L_j$ if and only if $pc_j = 0$. 
Since the path cost calculation contains only linear operations, 
we design a $\sf SRG$ algorithm to compute path costs via HSS. 
We further employ a random masking technique~\cite{wu2015} to hide classification labels, ensuring the client only learns the correct inference result $\mathcal{T}(\mathbf{x}) = v_j$ (corresponding to $pc_j = 0$) without gaining any information about other labels.

\noindent{\bf Verifiable Result Generation.}  
To enable the client to verify the correctness of the inference result $\mathcal{T}(\mathbf{x})$, 
we design a MAC-based verification mechanism and extend the $\sf SRG$ algorithm to a $\sf VRG$ algorithm. 
The mechanism consists of three key steps:
(i) MAC Key Distribution: The client generates a MAC key $A$, encrypts $A$ using HSS to produce ciphertext $\C_A$, and distributes $\C_A$ to the two servers. The security of HSS ensures that the servers cannot extract $A$ from $\C_A$, keeping $A$ secret to the client. 
(ii) Proof Generation: During server-side computation, we extend $\sf SRG$ with a verifiable label generation step. 
For each leaf node $L_j$, 
the servers compute a proof $w_j$ such that 
$w_j = A \cdot v_j$ (where $v_j$ is $L_j$'s classification label)
via HSS, using the encrypted MAC key $\C_A$ and encrypted label $\C_{v_j}$. 
(iii) Result Verification: Upon receiving proofs from the servers, the client verifies the correctness of $\mathcal{T}(\mathbf{x}) = v_j$ by checking if $A \cdot v_j = w_j$.
Such MAC-based verification mechanism is secure against a single malicious server, and is widely used to construct maliciously secure MPC protocols~\cite{Damagrd2012}. 
For collusion-resistant security, combining HSS with ZKPs is a promising direction~\cite{Choudhuri2024}, and left for future work.

\section{Related Work}
Existing ODTE protocols primarily rely on FHE \cite{lu2018,xu2023,aloufi2021} 
and SS \cite{liu2019,ma2021,zheng2023,zheng2022,tsuchida2020,ji2023,cheng2024}. 
The optimal choice depends on the network environments, as these primitives often have a trade-off regarding the number of required cloud servers, security guarantees, computation cost and communication cost. 
Our focus is to construct a two-server ODTE protocol that is free of S2S rounds, 
and achieves security against \emph{malicious} servers, 
which is suitable in WAN settings.

Several studies \cite{lu2018,aloufi2021}  
have sought to construct single-server ODTE protocols, 
where both the encrypted decision tree model and the encrypted client's input data 
are sent to a single server. 
While these protocols eliminate the need for S2S rounds by requiring only a single cloud server, 
they require the full power of FHE, 
which incurs significant computational overhead.
Additionally, these protocols fail to protect the client's feature indices from the cloud servers, thereby allowing the server to identify the type of data evaluated at each decision tree node.
Such exposure enables the cloud servers to infer more about the decision tree model \cite{kiss2019}, 
compromising the intellectual property of the model provider.

As opposed to the single-server approach, 
another line of research has proposed two-server ODTE protocols 
to reduce the computation overhead by distributing the computations among two servers \cite{zheng2022,xu2023,liu2019,ma2021,zheng2023}.
However, these protocols typically 
require $\mathcal{O}(h)$, $\mathcal{O}(\secpar)$ or $\mathcal{O}(\log (\secpar))$ S2S rounds,
where $h$ is the height of the decision tree and 
$\secpar$ is the security parameter.
Liu et al. \cite{liu2019} leveraged secret sharing and additively homomorphic encryption 
to construct an ODTE protocol with $h+3$ communication rounds 
by using Beaver's triplets for multiplication. 
Ma et al. \cite{ma2021} constructed an ODTE protocol 
by combining heavy techniques like key management, conditional Oblivious Transfer (OT) and garbled circuits, 
which necessitates $2h$ S2S rounds.
Xu et al. \cite{xu2023} distributed single-server protocol \cite{lu2018} to 
two unbalanced servers by using multi-key FHE schemes.
In their protocol, one of the servers holds the decryption key 
which can learn the comparison result at each level of the decision tree, 
resulting in $2(h+1)$ S2S rounds 
but more information leakage compared with others. 
Zheng et al. \cite{zheng2022} achieved more lightweight computation by 
replacing the additively homomorphic encryption approach of Tai et al. \cite{tai2017} with 
additive-sharing counterparts (and Beaver's triplets) 
using known secure computation tricks \cite{decock2017} in ideal network settings. 
However, they require $2(\secpar-1)$ S2S rounds. After the optimizations in \cite{zheng2023}, 
they still require $\log (\secpar)+5$ S2S rounds.

Recently, several three-server ODTE protocols have been proposed 
to achieve constant S2S rounds using 
replicated secret sharing (RSS) schemes \cite{tsuchida2020,cheng2024} and 
function secret sharing (FSS) schemes \cite{ji2023}. 
Tsuchida et al. \cite{tsuchida2020} introduced the first 
constant-round ODTE protocol, requiring 25 communication rounds. 
Cheng et al. \cite{cheng2024} reduced this to 5 rounds, and 
Ji et al. \cite{ji2023} further optimized it to 4 rounds. 
However, Cheng et al. \cite{cheng2024} claimed that 
protocol \cite{ji2023} 
achieves a weaker functionality than \cite{tsuchida2020,cheng2024}. Specifically, 
\cite{ji2023} only supports interval checks rather than 
general less-than comparisons and equality checks.

While these multi-server ODTE protocols aim to reduce S2S rounds, 
they still require {\em multiple S2S rounds} and assume that servers are \emph{semi-honest}.
Moreover, 
the protocol of \cite{liu2019} does not protect the feature indices from servers, and 
the protocol of \cite{xu2023} does not protect the comparison results from one of the servers. 
TABLE \ref{table:comp_theoretically} summarizes these ODTE protocols. 
In contrast, the two-server ODTE protocol in this paper 
protects both the feature indices and the comparison results from all servers. 
Furthermore, our protocol eliminates the need for S2S rounds
and achieves security against \emph{malicious} servers.

To achieve security against malicious servers, a common approach is to leverage general-purpose maliciously secure multi-party computation (MPC) protocols \cite{Damagrd2012} in place of the semi-honest SS schemes used in \cite{zheng2022,xu2023,liu2019,ma2021,zheng2023}. 
However, such generic transformations are not straightforward; they typically necessitate extensive additional rounds of S2S communication and incur heavy computational overhead, which is often unsuitable for latency-sensitive WAN settings.

In the context of non-outsourced private decision tree evaluation (PDTE), several studies have explored malicious security with different cryptographic techniques. For instance, Tai et al.~\cite{tai2017} and Wu et al.~\cite{wu2015} leverage conditional Oblivious Transfer (OT) and Proofs of Knowledge (PoK) to ensure the integrity of the evaluation process. More recently, Bai et al. \cite{bai2023} achieved malicious security by utilizing Consistent RSS, which detects inconsistencies among servers.
Despite these advances, many PDTE protocols~\cite{bost2015,cong2022, yuan2024,akhavan2023,cong2024}, remain focused on the semi-honest model to maintain efficiency. 
Note that, these protocols require both the client and the provider to have sufficient computational 
resource and stay online during the evaluation process.

\section{Preliminaries}
\label{sec:prel}

For any integer $n> 0$, we denote $[n] = \{1,\dots, n\}$. 
We denote by $\ZZ$ the set of all integers and by $\NN$ the set of all  non-negative integers.
Let $\secpar \in \NN$ be a security parameter. 
We say that a function $\epsilon(\secpar)$ is
{\em negligible}  in $\lambda$ and denote $\epsilon={\sf negl}(\lambda)$ if
$\epsilon(\lambda)=o(\lambda^{-c})$ for any constant $c>0$.
For any finite set $\calS$, we denote  by ``$s \leftarrow \calS$''  the process of
choosing an element $s$ uniformly
from $\calS$. For any algorithm $\mathcal{A}$, we denote by
``$y\leftarrow \mathcal{A}(x)$'' the process of running $\mathcal{A}$ on an input $x$ and assigning the output to
$y$.
We denote that
$({\langle x \rangle}_0,{\langle x \rangle}_1) \in\mathbb{N}^2$ is a {\em subtractive sharing} of $x$ if
${\langle x \rangle}_1 - {\langle x \rangle}_0 = x $. 
For a relation $R$, $(R)$ denotes $1$ if $R$ is true and $0$ otherwise. For example, $(\alpha>\beta)$ denotes $1$ if $a>b$ and $0$ otherwise.

\subsection{Homomorphic Secret Sharing}
\label{sec:HSS_Paillier}
Homomorphic secret sharing (HSS) \cite{boyle2016} allows two servers
(denoted by $\mathcal{S}_0,\mathcal{S}_1$) to
non-interactively evaluate functions on secret-shared data obtaining secret-shared results. 
The set of supported functions is usually restricted to a class $\mathcal{P}$. 
To some extent, HSS can be viewed as a distributed variant of homomorphic encryption but is roughly one order of magnitude faster.
Following  \cite{boyle2016}, an HSS scheme ${\sf HSS}=\mathsf{(Setup,Input,Eval)}$ for a class $\mathcal{P}$  of programs that have an input space
$\mathcal{I}$ and output space $R$ consists of three $\ppt$ algorithms
 with the following syntax:
\begin{itemize}
  \item $({\sf pk}, ({\sf ek_0,ek_1}))\leftarrow \mathsf{Setup}(\secparam)$:
  Given a security parameter $\secpar$, this algorithm generates a  {\em public key}
  $\pk$, which will be used by a client to encrypt its input, and a pair $(\ek_0,\ek_1)$ of
  {\em private evaluation keys}, which will be given to the servers  to perform  local computations.
  \item $\C_x\leftarrow \mathsf{Input}(\pk,x)$:
  Given   $\pk$ and a private input  $x\in \mathcal{I}$,
  this algorithm outputs a ciphertext $\C_x$.
  \item $\langle y \rangle_\sigma \leftarrow \mathsf{Eval}(b,\ek_\sigma,(\C_{x_1},\dots,\C_{x_\rho}),P) $:
  Given      a program $P \in \mathcal{P}$ and
  $\rho$ ciphertexts $(\C_{x_1},\dots,\C_{x_\rho})$, the $\sigma$th server ($\sigma\in \bin$) will
  use its evaluation key $\ek_\sigma$ to compute a share $\langle y \rangle_\sigma \in R$
  of the program's  output $y=P(x_1,\ldots,x_\rho) \in R$
  such that $\langle y \rangle_1 - \langle y \rangle_0 =y$.
\end{itemize}

An HSS scheme for a class $\mathcal{P}$ of programs
should satisfy the properties of correctness and security.
The property of {\em correctness} requires that
if all algorithms of $\sf HSS$ are correctly executed,
then $\langle y \rangle_1 - \langle y \rangle_0 =y$ with probability $1-\negl$.
The property of {\em security} requires that any $\ppt$ adversary that
controls a server (say the $b$th server $\calS_b$) cannot learn any information
about an input $x$ from ${\sf C}_x$.

\subsection{RMS Programs}
\label{sec:rms}

Many existing HSS schemes \cite{boyle2019,boyle2016,orlandi2021,abram2022} support the homomorphic evaluation of 
restricted multiplications straight-line (RMS) programs \cite{cleve1990}. 
RMS programs can represent arithmetic circuits with a specific constraint: multiplication gates must operate on an {\em input value} and a {\em memory value}. Within the HSS framework, for a plaintext $x$, an input value of $x$ refers to the HSS ciphertext $\C_x$, while a memory value of $x$ refers to the intermediate value $\M_{x,\sigma},\sigma\in \bin$. 
The HSS framework encompasses five fundamental instructions for RMS program evaluation, which are described as follows:
\begin{itemize}
  \item $\mathsf{ConvertInput}(\ek_\sigma,\C_x)$: 
  Given an evaluation key $\ek_\sigma$, 
  convert any ciphertext $\C_x$ into a memory value ${\sf M}_{x,\sigma}.$
  \item $\mathsf{Add}(\M_{x,\sigma},\M_{y,\sigma})$:
  Given two memory values $\M_{x,\sigma}$ and $\M_{y,\sigma}$,
  output a memory value $\M_{z,\sigma}$ for $z=x+y$.
  \item $\mathsf{Add}(\C_x,\C_y)$:
  Given two ciphertexts  $\C_x$ and $\C_y$,
  output a ciphertext $\C_z$ for $z=x+y$.
  \item $\mathsf{Mul}(\C_x,\M_{y,\sigma})$:
  Given a ciphertext $\C_x$ and a memory value $\M_{y,\sigma}$,
  output a memory value $\M_{z,\sigma}$ for $z=x\cdot y$.
  \item $\mathsf{Output}(\M_{x,\sigma},n_{\mathsf{out}})$:
  Given a memory value $\M_{x,\sigma}$ and an output modulus $n_{\mathsf{out}}$,
  output ${\langle x \rangle}_\sigma \bmod n_{\sf out}$.
\end{itemize}

Given that RMS inherently supports addition operations between two memory values, we extend this instruction set with two additional instructions: subtraction (${\sf HSS.Sub}$) between two memory values and constant multiplication (${\sf HSS.cMul}$), which can be described as follows: 
\begin{itemize}
  \item $\mathsf{Sub}(\M_{x,\sigma},\M_{y,\sigma})$:
  Given two memory values $\M_{x,\sigma}$ and $\M_{y,\sigma}$,
  output a memory value $\M_{z,\sigma}$ for $z=x-y$.
  \item $\mathsf{cMul}(c,\M_{y,\sigma})$:
  Given an integer $c$ and a memory value $\M_{y,\sigma}$,
  output a memory value $\M_{z,\sigma}$ for $z=c\cdot y$.
\end{itemize} 
These supplementary instructions play a crucial role in our ODTE protocol.

\section{Problem Statement}
\label{sec:pbst}

\begin{figure}[t]
  \centering
  \includegraphics[width=\columnwidth,height=5cm]{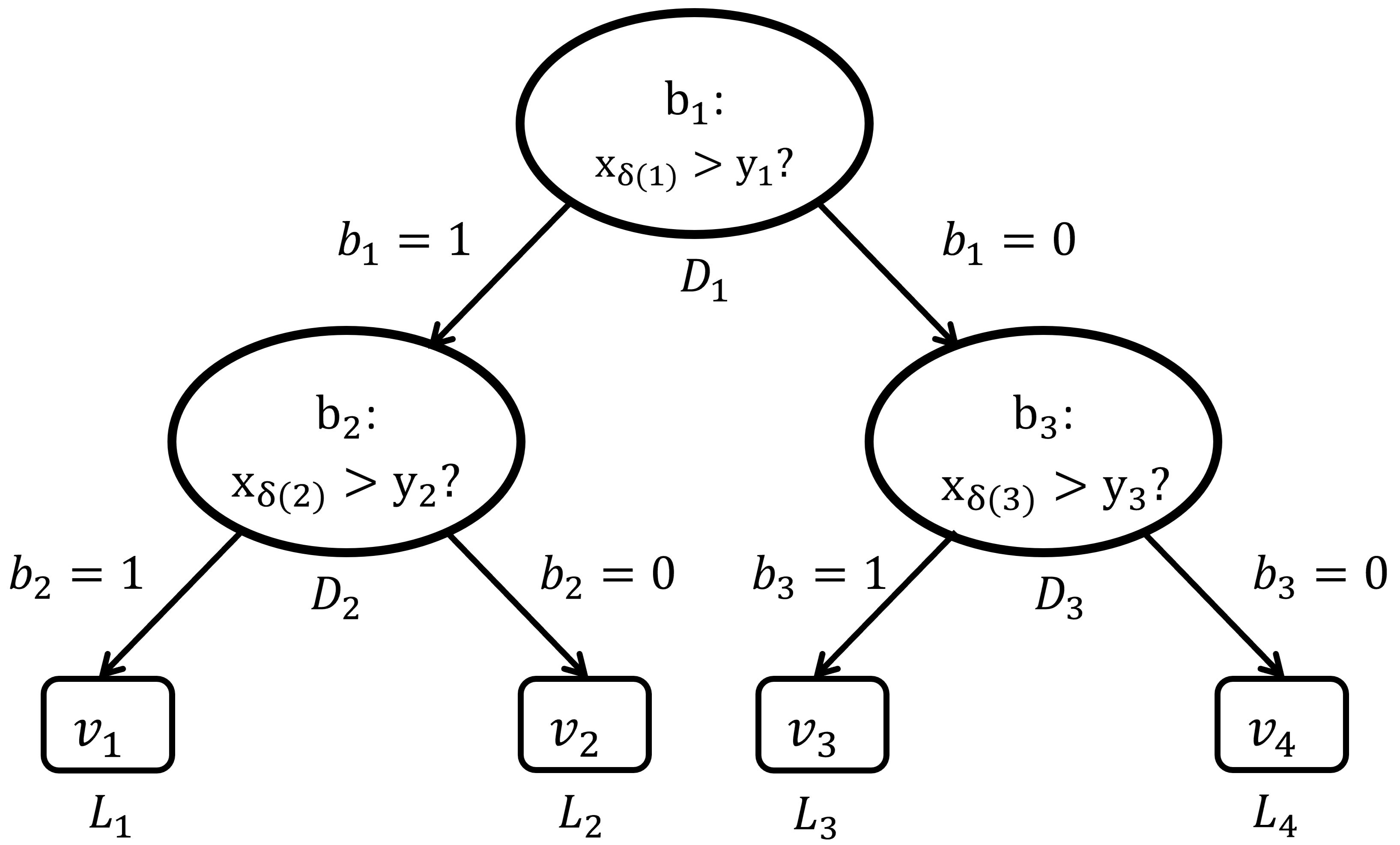}
\caption{A decision tree of height 2 ($m=3,k=4$)}
\label{fig:Tree}
\vspace{-4mm}
\end{figure}

In this section, 
we provide the background on decision tree evaluation and then describe the system framework, and the threat model for our ODTE protocol.

\subsection{Background on Decision Tree Evaluation}
\label{sec:DTE}
A {\em  decision tree} (see {\bf Fig. \ref{fig:Tree}} for an example) is a binary tree together with two vectors:
  a vector ${\bf y}=(y_1,\dots,y_{m}) \in \ZZ^m$
of {\em threshold  values}
and a  vector ${\bf v}=(v_1,\dots,v_k) \in  \ZZ^k$ of {\em classification  labels}.
Each {\em non-leaf  node} $D_j (1\leq j\leq m)$ of the tree is called  a  decision node  and associated with  a threshold value $y_j$;
each {\em leaf node} $L_j (1\leq j\leq k)$ of the tree  is associated with a classification label $v_j$.
A decision tree  takes a  {\em feature vector}  ${\bf x} = (x_1,\dots,x_n) \in \ZZ^n$ as input
and  outputs a   classification label.
Given a feature vector $\bf x$ and a mapping $\delta: [m]\rightarrow  [n]$
 as input, 
the evaluation of a decision tree is done by traversing a path
from the tree's root to a leaf as follows:
{\em
Performing feature selection by assigning a feature $x_{\delta (j)}$ 
to each decision node $D_j$ for all $j\in [m]$. 
To evaluate the tree, starting from the root $D_1$  of the tree,  
evaluate the boolean testing function
$f(x_{\delta (j)},y_j)=(x_{\delta (j)} > y_j)$ at every decision node $D_j$ to get a bit $b_j\in \{0,1\}$;
traverse the left  branch if  $b_j=1$ and  traverse the right branch if
  $b_j=0$, until a leaf node is reached.}

  \begin{figure}[t]
  \centering
  \includegraphics[width=\columnwidth,height=5cm]{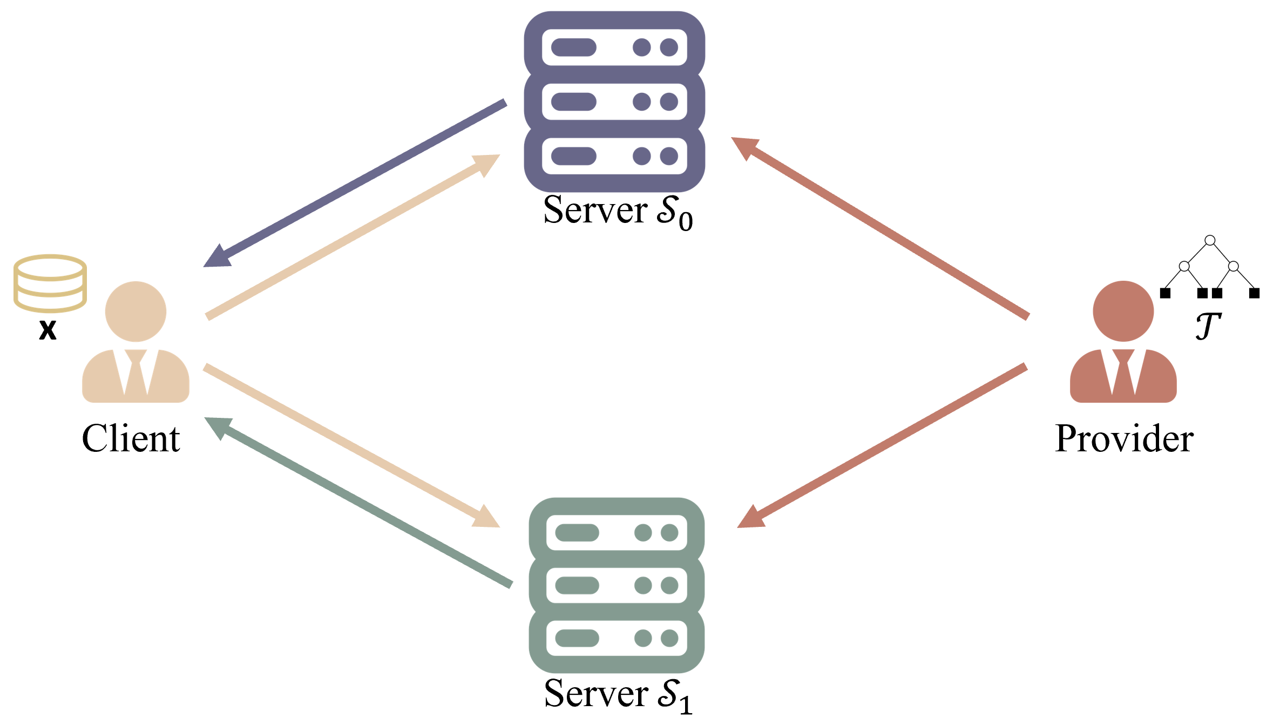}
\caption{Our ODTE System Framework}
\label{fig:Model}
\vspace{-4mm}
\end{figure}

\subsection{System Framework}

Our system framework for the ODTE protocol (see {\bf
Fig. \ref{fig:Model}}) consists of three types of entity: a
{\em model provider}, a {\em client}, and two {\em cloud servers}.

\noindent
{\bf Model provider.}
In our system, the model provider owns a proprietary decision tree
$\mathcal{T}$ 
(including the
 vector ${\bf y}$ of  threshold  values, the
  vector ${\bf v}$ of  classification  labels 
  and the mapping $\delta$) 
 that can make classification on
 a feature  vector $\bf x$. 
The model provider aims to offer classification services to clients by deploying the decision tree model on cloud servers, 
which provides benefits such as faster response times and ease of management.
In our framework, the model provider encrypts the threshold values ${\bf y}$ and classification labels ${\bf v}$ before sending them to each cloud server.
The encrypted feature mapping $\delta$ is made publicly available, after which the model provider can go offline.

 \noindent
{\bf Client.}
The client possesses a feature vector $\bf x$, 
which contains private and sensitive information.
The client seeks to leverage the computational powers of the cloud servers
to obtain the classification result ${\cal T}({\bf x})$. 
Initially, the client downloads the encrypted mapping $\delta$, 
a one-time cost since $\delta$ remains unchanged across multiple inferences. During the online phase, the client encodes the feature vector $\bf x$ and sends it to both cloud servers.
In response, the client receives \textit{a single response} message from each cloud server.
Using these two responses, the client efficiently extracts the classification result $\mathcal{T}({\bf x})$.

\noindent
{\bf Cloud servers.}
In our system,  the two cloud servers $\mathcal{S}_0$ and $\mathcal{S}_1$
are hosted by two independent cloud service providers.
Upon receiving the client's request,
each server  performs the required computations locally  and generates
{\em a single response} message to the client.
Notably, we require the cloud servers {\em never communicate with each other}.
This key distinction from existing two-server models for ODTE \cite{liu2019,zheng2022,zheng2023,ma2021,xu2023} significantly reduces the risk of collusion between servers and eliminates the high latency associated with S2S rounds.

\subsection{Threat Model}
\label{sec:threat model}

As in prior ODTE protocols, we first consider \emph{semi-honest} adversaries. These adversaries are assumed to follow protocol specifications but may attempt to extract additional information through local computations.
We assume a static corruption strategy, where the adversary corrupts one party (client or server) before protocol execution and gains access to its input.

In addition, our ODTE protocol is designed to be secure 
against \emph{malicious} servers. 
A corrupted server may attempt to access the proprietary decision tree $\mathcal{T}$ or the client's feature vector $\bf x$ and may tamper with the classification result $\mathcal{T}({\bf x})$ to mislead the client.
Our system ensures the following security guarantees: (i) A single cloud server cannot infer information about either $\bf x$ or $\mathcal{T}$; (ii) the client can verify the integrity of the classification result $\mathcal{T}({\bf x})$.

\noindent{\it Remark.} 
To hide the  structure of the tree $\mathcal{T}$, several
existing works \cite{wu2015,tai2017,zheng2023,zheng2022,xu2023,ma2021,liu2019}
have suggested the
non-complete decision trees should be made complete
by adding dummy decision nodes.
Without loss of generality, we  will
focus on {\em complete binary decision trees} in our protocol.

\subsection{Non-colluding Assumption}
To achieve privacy in multi-server protocols, secret sharing serves as a foundational technique: 
a user splits its private data $x$ into $n$ shares $\{x_1, x_2, ..., x_n\}$ and distributes them to $n$ servers. 
An $n$-server secret-sharing based protocol typically guarantees $t$-privacy (where $t<n$): 
meaning that any subset of $t$ servers cannot reconstruct $x$ from their shares, 
while $t+1$ servers can.
For two-server protocols ($n=2$), the $t$-privacy 
reduces to the non-colluding assumption ($t=1$): 
A single server cannot reconstruct $x$ from its share, but two servers can. 
Therefore, the non-colluding assumption is inherent to preserving privacy in two-server ODTE protocols.
To our knowledge, all existing two-server ODTE protocols \cite{liu2019,ma2021,xu2023,zheng2023,zheng2022,tsuchida2020,ji2023,cheng2024} rely on the non-colluding assumption.

To ensure verifiability against colluding servers, a common approach requires all servers to jointly participate in a zero-knowledge proof system \cite{groth2016a}. 
This ensures that even if all servers collude, they cannot manipulate the output without being detected. 
We defer exploring this promising direction to future work.

While the non-colluding assumption is inherent, several practical measures 
can significantly reduce the risk of collusion. 
First, our non-interactive protocol design relaxes the practical strictness 
of this assumption. Specifically, a model provider can deploy their models 
on two distinct cloud platforms (e.g., AWS and Azure), with no 
requirement for either cloud to be aware of the other's 
involvement~\cite{zhang2022a}. 
Second, game-theoretic mechanisms can be employed to actively deter and 
detect collusion between servers. 
Gong et al.~\cite{gong2024} proposed a game-theoretic framework that 
incentivizes rational servers to behave honestly by making collusion 
economically irrational, providing an additional layer of practical 
assurance beyond infrastructure-level separation.

\section{Secure Integer Comparison}
\label{sec:sic}
As described in Section \ref{sec:DTE}, 
the evaluation of a decision tree hinges on a sequence of conditional checks at each internal node: for a given feature value $x_{\delta(j)}$ (selected via the model's feature mapping $\delta$) and threshold $y_j$, determine if $x_{\delta(j)} > y_j$. 
This comparison operation is the core of decision tree inference—yet it is also the most privacy-sensitive step. 
A naive implementation would leak not only the feature value and threshold but also the structure of the decision tree (e.g., which paths are taken). 
For ODTE, ensuring the privacy of this comparison while avoiding S2S rounds is a critical challenge.

Existing secure integer comparison solutions for multi-server settings~\cite{Damagrd2006,Nishide2007,Reistad2007,Sutradhar2023,zheng2022,zheng2023} rely on interactive protocols (e.g., rounds of message exchange between servers) to achieve privacy. 
However, these protocols are ill-suited for our non-interactive ODTE framework: frequent S2S communication introduces heavy communication overhead and high network latency (especially in WANs). 
A summary of the S2S rounds, communication cost in bits, and the number of multiplication gates required by these solutions is provided in TABLE \ref{table:sic}.

To address this gap, we design the first non-interactive secure integer comparison ($\sf SIC$) algorithm for two servers, which enables privacy-preserving comparison of integers encrypted by HSS without any S2S rounds. 
Our algorithm leverages HSS to support distributed, local computation, and outputs an expressive result (a secret share of the comparison bit) that can be directly used in subsequent decision tree traversal steps.

\subsection{Comparison of Two Integers within HSS}
Given two integers $\alpha$ and $\beta$ in the clear, computing $(\alpha>\beta)$ can be achieved through several established methods \cite{fischlin2001,schoenmakers2004,Kerschbaum2006}. 
The conventional approach involves calculating their difference $\alpha-\beta$ and determining its sign bit. 
However, this approach is insufficient in our non-interactive setting, where each server only holds shares of $\alpha-\beta$, making it impossible to directly compute the sign bit without S2S rounds. 

To address this limitation, we adopt an arithmetic circuit-based formulation. Let: 
$\alpha=\sum_{i=1}^t 2^{i-1} \cdot \alpha_i$, $\beta=\sum_{i=1}^t 2^{i-1} \cdot \beta_i$,
where $\alpha_i,\beta_i\in \{0,1\}$. 
This representation enables bit-level comparison operations through two main approaches: the Most Significant Bit (MSB) method and the Least Significant Bit (LSB) method. The MSB approach identifies the first differing bit by examining bits from position $i=t$ to $i=1$, while the LSB method builds the comparison result progressively from $i=1$ to $i=t$. Each method offers distinct advantages in terms of circuit depth and implementation efficiency within various secure computation frameworks. 

In the circuit-based approach, basic boolean operations and comparisons can be expressed as arithmetic equations. Specifically, for any two bits 
$\alpha_i,\beta_i\in \{0,1\}$, the boolean operations ($\wedge,\vee$) and comparison relations ($>,=$) are expressed in terms of arithmetic operations over integers $\ZZ$ as follows:
\begin{equation}
  \label{eq:gt}
\begin{aligned}
  (\alpha_i \wedge  \beta_i)&:=\alpha_i \cdot \beta_i, \\
  (\alpha_i\vee \beta_i)&:=\alpha_i+\beta_i-\alpha_i \cdot \beta_i, \\ 
  (\alpha_i>\beta_i)&:=\alpha_i \cdot(1-\beta_i),  \\
  (\alpha_i=\beta_i)&:=1-\alpha_i-\beta_i+2\cdot \alpha_i \cdot \beta_i. 
  \end{aligned}
\end{equation}

\noindent{\bf MSB approach}. 
The MSB approach relies on identifying the leftmost differing bit position. However, this approach is incompatible with non-interactive settings where servers only hold shares of the bits. Instead, the comparison can be rephrased as identifying the existence of a bit position $i$ where 
$\alpha_i>\beta_i$ while all more significant bits remain equal. 
This can be expressed as: $
  (\alpha>\beta) := \bigvee_{i=1}^{t}[(\alpha_i>\beta_i)\bigwedge_{j=i+1}^{t}(\alpha_i=\beta_i)]$. 
Using the arithmetic expressions in Eq. \eqref{eq:gt}, the MSB approach can also be formulated as:
\begin{displaymath}
  (\alpha>\beta) := \sum_{i=1}^{t}[\alpha_i(1-\beta_i)\prod_{j=i+1}^{t}(1-\alpha_i-\beta_i+2\alpha_i \beta_i)].
\end{displaymath}
This yields a $2t$-variate polynomial $P_{(\alpha>\beta)}$ of degree $2t$, 
demanding $t(t+1)/2$ homomorphic multiplications \cite{chen2023}. Each server $\calS_\sigma$ can evaluate $P_{(\alpha>\beta)} $ locally by executing ${\sf HSS.Eval}(\sigma, \ek_\sigma, (\{\C_{\alpha_i}\}_{i=1}^t, \{\C_{\beta_i}\}_{i=1}^t),P_{(\alpha>\beta)})$.

\begin{algorithm}[t]
  \caption{$\M_{c_{t},\sigma}\leftarrow {\sf SIC}(\sigma, \ek_\sigma, \{\C_{\alpha_i}\}_{i=1}^t, \{\C_{\beta_i}\}_{i=1}^t)$}  \label{alg:sic}
\begin{algorithmic}
\State  {\bf Input}. (1)
    $\sigma\in \{0,1\}$: the index of a  server ${\cal S}_\sigma$;
  (2) $\ek_\sigma$: the private evaluation key of HSS;
  (3) $\C_{\alpha_i}$ (resp. $\C_{\beta_i}$): a ciphertext  of the bit  $\alpha_i$ (resp. $\beta_i$)
  under  ${\sf HSS}$, where $1\leq i\leq t$.
\State {\bf Output}.  $\M_{c_{t},\sigma}$: a memory value of the bit $c_{t}$.
\State {\bf Step 1}. $\calS_\sigma$ generates a trivial memory value ${\sf M}_{1,\sigma}$;
\State {\bf Step 2}. ${\sf M}_{\alpha_1,\sigma} \leftarrow {\sf HSS.ConvertInput}(\ek_\sigma,\C_{\alpha_1})$;
\State  {\bf Step 3}. ${\sf M}_{c_1,\sigma} \leftarrow {\sf HSS.Mul}(\C_{\beta_1},\M_{\alpha_1,\sigma})$;
\State  {\bf Step 4}. ${\sf M}_{c_1,\sigma} \leftarrow {\sf HSS.Sub}(\M_{\alpha_1,\sigma},\M_{c_1,\sigma})$;
\State  {\bf Step 5}.
{\bf for $i\in \{1,\ldots,t-1\}$ do}
\State \hspace{15mm}
  ${\sf M}_{\alpha_{i+1},\sigma} \leftarrow {\sf HSS.ConvertInput}(\ek_\sigma,\C_{\alpha_{i+1}})$;
\State \hspace{15mm}
  $\C_{temp} \leftarrow {\sf HSS.Add}(\C_{\alpha_{i+1}},\C_{\beta_{i+1}} )$;
  \State \hspace{15mm}
  ${\sf M}_{c_{i+1},\sigma} \leftarrow {\sf HSS.Mul}(\C_{temp},{\sf M}_{c_i,\sigma})$;
  \State \hspace{15mm}    ${\sf M}_{c_{i+1},\sigma} \leftarrow {\sf HSS.Sub}({\sf M}_{c_i,\sigma}, \M_{c_{i+1},\sigma})$;
\State \hspace{15mm}    ${\sf M}_{2c_i,\sigma} \leftarrow {\sf HSS.Add}({\sf M}_{c_i,\sigma}, {\sf M}_{c_i,\sigma})$;
\State \hspace{15mm}    ${\sf M}_{2c_i-1,\sigma} \leftarrow {\sf HSS.Sub}({\sf M}_{2c_i,\sigma}, {\sf M}_{1,\sigma})$;
\State \hspace{15mm}    ${\sf M}_{temp,\sigma} \leftarrow {\sf HSS.Mul}(\C_{\alpha_{i+1}},{\sf M}_{2c_i-1,\sigma})$;
\State \hspace{15mm}    ${\sf M}_{temp,\sigma} \leftarrow {\sf HSS.Mul}(\C_{\beta_{i+1}},{\sf M}_{temp,\sigma})$;
\State \hspace{15mm}   ${\sf M}_{c_{i+1},\sigma} \leftarrow {\sf HSS.Add}({\sf M}_{c_{i+1},\sigma}, {\sf M}_{temp,\sigma})$;
 \State \hspace{15mm}  ${\sf M}_{c_{i+1},\sigma} \leftarrow {\sf HSS.Add}({\sf M}_{c_{i+1},\sigma}, {\sf M}_{\alpha_{i+1},\sigma})$;
\State {\bf Step 6}. Output ${\sf M}_{c_{t},\sigma}$.
\end{algorithmic}
\end{algorithm}

\noindent{\bf LSB approach}. 
The LSB approach iteratively builds the comparison result for two integers. Let $c_0=0$, and for $1\leq i\leq t$, let $c_i$ denote the comparison result for the $i$-bit substring 
$(\alpha_{i}\dots \alpha_1 > \beta_{i}\dots \beta_1)$. 
The value of $c_{i+1}$ can be recursively defined as:
\begin{equation}
c_{i+1}:=\  (\alpha_{i+1}>\beta_{i+1})\vee[(\alpha_{i+1}=\beta_{i+1})\wedge c_i], 0\leq i < t, \label{eq:lsbb}
\end{equation}
where $c_{t}$ indicates the final comparison result $(\alpha>\beta)$ 
for $\alpha,\beta$ being $t$-bit integers. 
Using the arithmetic expressions defined in Eq. \eqref{eq:gt}, we can transform Eq. \eqref{eq:lsbb} into:
\begin{equation}
  \begin{aligned}
  c_{i+1}:=\ & \alpha_{i+1}(1-\beta_{i+1}) + \\ & 
  c_{i} (1-\alpha_{i+1}-\beta_{i+1}+2\alpha_{i+1}\beta_{i+1}), 0\leq i < t,\label{eq:lsb}
  \end{aligned}
\end{equation} 
In this formulation, $(\alpha>\beta):=c_{t}$. 
This sequential formulation requires only $5t-4$ multiplication gates. 

To implement this LSB method in a non-interactive setting, Eq. \eqref{eq:lsb} must be compatible with the evaluation under HSS. 
Indeed, the multiplication operations in Eq. \eqref{eq:lsb} naturally align with the RMS program requirements: $c_i$ serves as an intermediate value (a memory value), while $\alpha_i,\beta_i$ are provided as HSS ciphertexts (input values). 
The product $\alpha_{i}\beta_i$ is computed by converting $\alpha_{i}$ to memory values using $\mathsf{HSS.ConvertInput}$. 
This HSS-based implementation requires a total of $6t-4$ multiplication gates during evaluation, accounting for the additional $\sf HSS.mul$ gate per iteration during $\sf HSS.ConvertInput$ (See Section \ref{sec:rms}).

To further optimize the LSB method, we can minimize the number of multiplication gates by transforming Eq. \eqref{eq:lsb} into a more efficient form. 
Through ingenious algebraic manipulation, we reconfigure and rewrite Eq. \eqref{eq:lsb}, converting some complex multiplication and addition gates into simpler ones:
\begin{equation}
  \begin{aligned}
  c_{i+1}:=\ & c_i - c_i(\alpha_{i+1}+\beta_{i+1}) + \\ & 
  (2c_{i}-1)\alpha_{i+1}\beta_{i+1}+\alpha_{i+1}, 0\leq i < t.\label{eq:flsb}
  \end{aligned}
\end{equation} 
The advantage of this optimization (Eq. \eqref{eq:flsb}) lies in the substantial reduction of computational complexity. 
Specifically, the number of multiplication gates is decreased to $4t-2$. 
With these optimizations implemented, we proceed to present our complete \textsf{SIC} algorithm ({\bf Algorithm \ref{alg:sic}}).

In {\bf Algorithm \ref{alg:sic}}, server ${\cal S}_\sigma$ first initializes the constant unit share ${\cal M}_{1,\sigma}$ in Step 1. and immediately computes the base term ${\cal M}_{c_1,\sigma}$ by effectively evaluating $\alpha_1(1 - \beta_1)$ on memory values in Step 2-4. 
Typically following the recursion in Eq. \eqref{eq:flsb}, ${\cal S}_\sigma$ iteratively updates its state to obtain final memory value ${\cal M}_{c_{t},\sigma}$ in Step 5. 
It should be noted that each server $\calS_\sigma$ can further extract a linear share of $c_{t}$ through the invocation of the algorithm $\langle c_{t} \rangle_\sigma \leftarrow \mathsf{HSS.Output}(\M_{c_{t},\sigma},n_{\sf out})$, where $n_{\sf out}$ denotes the output modulus, satisfying the relation 
$c_{t}=(\langle c_{t} \rangle_1  - \langle c_{t} \rangle_0)$. 

The security of our \textsf{SIC} algorithm originates from the underlying HSS scheme: each server ${\cal S}_\sigma$ has access only to the encrypted bits of $\alpha$ and $\beta$, and obtains exclusively the shares of the comparison result $c_{t}$. 
Due to the security properties of HSS, these encrypted bits and shares reveal no information regarding the actual values under comparison.
Consequently, our \textsf{SIC} algorithm successfully implements secure integer comparison while maintaining privacy of both the input values $\alpha,\beta$ and the comparison result $(\alpha>\beta)$ against
each server ${\cal S}_\sigma$. 

\subsection{Dealing with Floating Point Numbers}
\label{sec:float}
While the \textsf{SIC} algorithm is designed for integer arithmetic, decision tree models usually require floating point number comparisons \cite{bost2015}. Thus, we must adapt our algorithm.

Since \textsf{SIC} operates at the bit level, floating-point comparisons transform effortlessly into integer operations via constant scaling. By multiplying $\alpha$ and $\beta$ by $K$ (e.g., $K=2^{52}$ for IEEE 754 doubles), 
the inequality $\alpha > \beta$ maps to an integer comparison of $\alpha' = K\cdot \alpha$ vs. $\beta' = K\cdot \beta$. This strategy preserves evaluation accuracy under the algorithm's integer arithmetic.

The precision requirements for such comparisons are well-supported by the underlying HSS scheme, which is typically implemented in cryptographic frameworks such as RSA groups \cite{orlandi2021,abram2022} (e.g., Paillier). These HSS schemes provide an exponentially large plaintext space of approximately $2^{1024}$, allowing for the use of a sufficiently large $K$ without sacrificing numerical precision.
Consequently, our \textsf{SIC} algorithm ({\bf Algorithm \ref{alg:sic}}) securely extends to floating-point comparisons.

\section{Our Semi-honest ODTE Protocol}
\label{sec:prot}

In this section, we first propose a new semi-honest ODTE protocol based on 
the HSS scheme \textsf{HSS} in Section \ref{sec:HSS_Paillier}. 
The maliciously secure ODTE protocol is presented in Section \ref{sec:prot_malicious}.

\subsection{Protocol Overview}
At a high level, our protocol  consists of four phases: 
{\em setup},
{\em input preparation}, 
{\em server-side computation}, and {\em result reconstruction.}
We introduce each phase as follows. 
(i) {\bf Setup phase.} This phase initializes the protocol with HSS key pairs and a random permutation.
(ii) {\bf Input preparation phase.} This phase allows the provider to encrypt the decision tree, and the client to encrypt the feature vector.
(iii) {\bf Server-side computation phase.} In this phase, each cloud server locally evaluates the encrypted decision tree on the encrypted feature vector to generate a partial classification result for the client.
(iv) {\bf Result reconstruction phase.} This phase enables the client to recover the classification result corresponding to her feature vector.
We give the details of each phase in the subsequent sections.

\subsection{Setup Phase}
Our ODTE protocol begins with a trusted party executing the algorithm ${\sf HSS.Setup(\secparam)}$ (Section \ref{sec:HSS_Paillier}), 
which generates a public key $\pk$, 
and two private evaluation keys $\ek_0,\ek_1$,
as discussed in Section \ref{sec:HSS_Paillier}. 
The public key $\pk$ is available to all parties, including the client, the model provider and the two servers. 
Each private evaluation key $\ek_\sigma (\sigma=0,1)$ 
is only available to the server $\calS_\sigma$ respectively. 
Additionally, the trusted party selects a pseudo-random permutation (PRP) 
$\pi:[k]\rightarrow [k]$ and distributes $\pi$ to both servers.
Importantly, the requirement for a trusted setup can be eliminated by 
constructing the HSS scheme from Paillier-ElGamal encryption \cite{orlandi2021} 
and generating PRP by established PRP generators \cite{akl1984}.

\begin{algorithm}[t]
  \caption{ $\C_{\mathcal{M} \mathbf{x}}
\leftarrow {\sf SFS}( \C_{\mathcal{M}}, {\bf x} )$}  \label{alg:sfs}
\begin{algorithmic}
	\State  {\bf Input}. (1)
      $\C_{\mathcal{M}}$: an $m\times n$ matrix, $\C_{\mathcal{M}_{j,s}}$ is 
      the ciphertext of $\mathcal{M}_{j,s}$ under $\sf HSS$;
    (2) $\bf x$: the client's feature vector of length $n$, 
    each feature is $t$-bit long.
	\State {\bf Output}.  $\C_{\mathcal{M}\mathbf{x}}$: 
  an $m\times t$ matrix, 
  $\C_{(\mathcal{M}x)_{j,i}}$ is 
      the ciphertext of the $i$th bit of feature $x_{\delta (j)}$ 
      under $\sf HSS$.
  \State {\bf Step 1}.  
  Represent the feature vector $\bf x$ in binary format, 
  where ${\bf x}_s=\sum_{i=1}^{t} 2^{i-1} \mathbf{x}_{s,i}$.
\State  {\bf Step 2}.
{\bf for $j\in\{1,\dots,m\}$ do}
\State \hspace{15mm}
{\bf for $i\in\{1,\dots,t\}$ do}
\State \hspace{20mm}
$\C_{(\mathcal{M}x)_{j,i}} \leftarrow \sum_{s=1}^{n} \C_{\mathcal{M}_{j,s}\cdot { x}_{s,i} }$
\State \Comment $
\C_{\mathcal{M}_{j,s}\cdot { x}_{s,i} }:= \C_{\mathcal{M}_{j,s}} 
{\rm \ if\ } \mathbf{x}_{s,i}=1 {\rm \ else\ } 
\C_{\mathcal{M}_{j,s}\cdot {x}_{s,i} }:= \C_0$;
\State {\bf Step 3}.  Output  $\C_{\mathcal{M}\mathbf{x}}$.
\end{algorithmic}
\end{algorithm}

\subsection{Input Preparation Phase}
\label{sec:ipp}
In this phase, the model provider and the client prepare their respective inputs, 
${\cal T}$ and $\bf x$, for the cloud servers. 

\noindent\textbf{Model provider.} The model provider first encrypts the threshold value $y_j$ at each decision node of the decision tree $\cal T$ and the classification label $v_i$ at each leaf node of $\cal T$: 
\begin{itemize}
    \item For $j\in [m]$, encrypt the threshold value
     ${y}_j=\sum_{i=1}^t 2^{i-1} \cdot {y}_{j,i}$ at each
      decision node $D_j$  by bit using  $\sf HSS$:
    \begin{equation} \label{eqn:cyji}
      \C_{{y}_{j,i} } \leftarrow {\sf HSS.Input}(\pk,{ y}_{j,i}). 
    \end{equation}
    \item For $i\in [k]$, encrypt the classification  label ${v}_i$
    at each leaf node $L_i$ by bit using   $\sf HSS$:
    \begin{equation} \label{eqn:cvi}
      \C_{{v}_i} \leftarrow {\sf HSS.Input}(\pk,{v}_i).
    \end{equation}
\end{itemize}
In Eq. \eqref{eqn:cyji}, we encrypt the threshold values by bit 
because our \textsf{SIC} algorithm ({\bf Algorithm \ref{alg:sic}}) 
takes the HSS ciphertexts of bits as input. 
Besides to the encryption of these data values, 
the provider must also appropriately encrypt the mapping 
$\delta: [m] \rightarrow [n]$ for use in feature selection while maintaining privacy.
To solve this problem, our main idea is to find an appropriate representation for $\sigma$, 
and we resort to representing $\sigma$ as an $m\times n$ matrix $\mathcal{M}$ 
where each row $\mathcal{M}_j$ is a binary vector with $n$ elements and 
the only nonzero element of $\mathcal{M}_j$ is located at the position $\delta(j)$. 
Under such a representation, it is easy to see that 
$x_{\delta{(j)}}=\mathcal{M}_j {\bf x}$. 
The model provider may encrypt the matrix $\mathcal{M}$ as $\C_{\mathcal{M}}$ via $\sf HSS.Input$:
  \begin{equation} \label{eqn:cmap}
    \C_{{\mathcal{M}}_{j,s}} \leftarrow {\sf HSS.Input}(\pk,\mathcal{M}_{j,s}),j\in [m], s\in [n],
  \end{equation} 
thus ensuring the mapping's privacy through the {\em security} of $\sf HSS$. 

\noindent\textbf{Client.} 
A straightforward way for the client to prepare her encrypted feature vector is to encrypt it bit-by-bit, similar to the model provider's approach. 
However, this approach may be impractical due to the limitations of most HSS, which only supports evaluations on RMS programs. 
For instance, consider a client encrypting a feature vector $\bf x$ bit by bit into 
\textsf{HSS} ciphertexts $\{\C_{x_{s,i}}\}$, where $s\in [n],i\in [t]$, and sending these ciphertexts to the two servers. 
Upon receiving the ciphertexts, server $\calS_\sigma$ would need to compute $x_{\delta{(j)}}=\mathcal{M}_j {\bf x}$ via the product $\C_{\mathcal{M}_{j,s}} \cdot \C_{x_{s,i}}$. However, the group-based \textsf{HSS} schemes~\cite{boyle2016,orlandi2021,abram2022} does not support ciphertext-ciphertext multiplication, as its underlying encryption scheme is only linearly homomorphic over a group (e.g., Paillier). 
Consequently, evaluating $\mathcal{M}_{j,s} \cdot x_{s,i}$ requires one operand to remain in plaintext. 
Therefore, we propose a client-side \textit{secure feature selection} (\textsf{SFS}) algorithm that encrypts the feature vector and performs feature selection simultaneously. 
This algorithm leverages the additive homomorphism inherent in the \textsf{HSS} scheme, ensuring compatibility with subsequent computations while maintaining the efficiency and security of the protocol.

In our algorithm, the encrypted mapping matrix $\C_{{\mathcal{M}}}$ is made publicly available by the model provider. 
Since $\C_{\mathcal{M}}$ remains constant for a given decision tree, the client can download it once and reuse it for multiple inferences.
The algorithm $\sf SFS$ takes two inputs: the encrypted mapping matrix 
$\C_{\mathcal{M}}$ of dimension $m\times n$, and the client's feature vector ${\bf x}=(x_1,\dots,x_n)$. The feature vector is represented as an $n \times t$ matrix through bit decomposition of each feature.
The algorithm outputs an $m \times t$ matrix $\C_{\mathcal{M}\mathbf{x}}$, 
where each entry $\C_{(\mathcal{M}\mathbf{x})_{j,i}}$ is an HSS ciphertext of the $i$th bit of $x_{\delta(j)}=\mathcal{M}_{j}{\bf x}$. 
This approach ensures that the client's feature vector is both feature-selected and encrypted bit by bit, satisfying the requirements for subsequent secure computation.

The \textsf{SFS} algorithm leverages the additive homomorphism of HSS ciphertexts, utilizing the $\C_{x+y}\leftarrow {\sf Add}(\C_x,\C_y)$ instruction (described in Section \ref{sec:rms}). 
Observing that: 
\begin{equation} \label{eqn:mx}
  \mathcal{M}_{j}{\bf x}=(
  \sum_{s=1}^{n} \mathcal{M}_{j,s}\cdot x_{s,1},\dots, 
  \sum_{s=1}^{n} \mathcal{M}_{j,s}\cdot x_{s,t}),
\end{equation}
where ${x}_{s,i}, i\in [t]$ is the $i$th bit of the feature $x_s$, 
it follows that: $\mathcal{M}_{j,s}\cdot x_{s,i}== \mathcal{M}_{j,s}$ if 
$x_{s,i}=1$ and $0$ otherwise.
With $\mathcal{M}_{j,s}$ encrypted as $\C_{\mathcal{M}_{j,s}}$, Eq. \eqref{eqn:mx} can be expressed as: 
\begin{equation} \label{eqn:mx2}
  \C_{\mathcal{M}_{j}{\bf x}}=(
  \sum_{s=1}^{n} \C_{\mathcal{M}_{j,s}\cdot x_{s,1}},\dots, 
  \sum_{s=1}^{n} \C_{\mathcal{M}_{j,s}\cdot x_{s,t}}).
\end{equation}
Here, 
$\C_{\mathcal{M}_{j,s}\cdot x_{s,i}}$ equals $\C_{\mathcal{M}_{j,s}}$ if $x_{s,i}=1$ and $\C_0$ (the HSS ciphertext of integer $0$) otherwise. 
The summation $\sum_{k=1}^{n}$ represents the aggregation of $n$ HSS ciphertexts. 
In practice, $\C_0$ need not be explicitly computed, as it can be omitted from addition operations. 
The complete implementation of the $\sf SFS$ algorithm is detailed in {\bf Algorithm \ref{alg:sfs}}. 
A quantitative analysis of the computational and communication costs of the $\sf SFS$ algorithm is provided in Section \ref{sec:offline cost}.

Although downloading $\C_{\mathcal{M}}$ may initially appear burdensome for the client, it is a one-time cost, as the matrix can be reused for all subsequent inferences with the same decision tree. A detailed quantitative analysis of the one-time offline costs is provided in Section~\ref{sec:offline cost}. 
For deployments where even this one-time cost is unacceptable, the supplementary material discusses an alternative instantiation based on RLWE-based \textsf{HSS}, which eliminates the offline download entirely by delegating the computation of $\C_{\mathcal{M}\mathbf{x}}$ to the servers, at the expense of increased server-side computational overhead.

\subsection{Server-side Computation Phase}
In this phase, each server $\calS_\sigma ~(\sigma=0,1)$ locally evaluates
  the encrypted decision tree $(\C_{\bf y}, \C_{\bf v})$ 
  on the encrypted selected feature vector
   $\C_{\mathcal{M}\mathbf{x}}$ to generate a partial classification result ${\rm output}_\sigma$
   for the client.
In particular, each server's computation requires two algorithms: 
  {\em  Secure integer comparison} (\textsf{SIC}); and
  {\em Secure  result generation} (\textsf{SRG}).

At each decision node $D_j, j\in [m]$,
each server takes both an encrypted  threshold value
$\{\C_{ y_{j,i}}\}, i\in [t]$ and an
encrypted feature value$\{\C_{(\mathcal{M}x)_{j,i}}\}, i\in [t]$ as input,
invokes the algorithm  \textsf{SIC} to compute a memory value of the Boolean  value
$b_j=({x}_{\delta (j)}>{ y}_j)$.
Once all decision nodes are evaluated,
the server invokes the algorithm \textsf{SRG}. At each decision node $D_j$, \textsf{SRG}
 assigns   a function of $D_j$'s share to every outgoing edge of
$D_j$. Finally, with the  values on all edges, \textsf{SRG}  computes
a share for each leaf node. 
The shares generated by the two servers, one from each server for every leaf node, collectively determine a \textit{value} that indicates whether the corresponding leaf node contains the classification result $\mathcal{T}({\bf x})$.

\begin{algorithm}[t]
  \caption{$ (\{ \langle pc_{\pi(i)}^* \rangle_{\sigma},\langle v_{\pi(i)}^* \rangle_{\sigma}
  \}_{i=1}^k )\leftarrow
{\sf SRG}(\sigma, \ek_\sigma,\pi$, 
$\{ \M_{b_j,\sigma} \}_{j=1}^m,
\{\C_{v_i} \}_{i=1}^k)$}\label{alg:srg}

\begin{algorithmic} 	
	\State  {\bf Input}.
    (1) $\sigma\in \{0,1\}$: the index of a  server ${\cal S}_\sigma$;
    (2) $\ek_\sigma$: the private evaluation key of HSS, which includes the secret key
    $\sf k_{prf}$ for the PRF;
    (3) $\pi$:  a random permutation of the integers $1,\ldots,k$;
    (4) $\{\M_{b_j,\sigma} \rangle\}_{j=1}^m$:
    the memory value of the $m$ comparison results at $m$ decision nodes;
    (5) $\{\C_{v_i}  \rangle\}_{i=1}^k$:
    the ciphertexts of the $k$ classification labels at $k$ leaf nodes.
	\State {\bf Output}.
  (1) $\{ \langle pc_{\pi(i)}^* \rangle_{\sigma}  \}_{i=1}^k $:
  the subtractive shares of $k$ randomized path costs;
  (2) $\{ \langle v_{\pi(i)}^* \rangle_{\sigma}  \}_{i=1}^k$:
  the subtractive shares of $k$ randomized classification labels.
  \State {\bf Step 1}. $\calS_\sigma$ generates a trivial memory value ${\sf M}_{1,\sigma}$;
  \State {\bf Step 2}.
   For each decision node $D_j, j\in [m]$, assign  the 
   memory value of $ec_{j,0}$ and $ec_{j,1}$ respectively  to the outgoing edges $E_{j,0}$ and 
    $E_{j,1}$: $
        \M_{ec_{j,0},\sigma} \leftarrow {\sf HSS.Sub}(\M_{1,\sigma}, 
        \M_{b_j,\sigma})$,  
        $\M_{ec_{j,1},\sigma} \leftarrow \M_{b_j,\sigma}$.

\State  {\bf Step 3}.  For each leaf node $L_i,i\in [k]$, set the path cost $\M_{pc_i,\sigma} = \sum_{E_{j,\ell}\in P_i} \M_{ec_{j,\ell},\sigma}$.

\State {\bf Step 4}.  For each leaf node $L_i,i\in [k]$, convert the HSS ciphertexts of classification labels into memory values: 
\begin{displaymath}
  \M_{v_i,\sigma} \leftarrow {\sf HSS.ConvertInput}(\ek_\sigma, \C_{v_i}).
\end{displaymath}
\State {\bf Step 5}. For every $i\in [k]$, generate two random integers 
$r_{i,0}=F_{\sf k_{prf}}(i\|0)$ and $r_{i,1}=F_{\sf k_{prf}}(i\|1)$
 using the PRF $F$. 
Compute the randomized path cost and the randomized classification label as
\begin{displaymath}
  \begin{aligned}
    \M_{pc_i^*,\sigma} & \leftarrow {\sf HSS.cMul}(r_{i,0}, \M_{pc_i,\sigma}), \\ 
\M_{v_i^*,\sigma} & \leftarrow {\sf HSS.Add}(\M_{v_i,\sigma},
{\sf HSS.cMul}(r_{i,1}, \M_{pc_i,\sigma})).
  \end{aligned}
\end{displaymath}
\State {\bf Step 6}. Convert the memory values in {\bf Step 5} into linear shares: 
\begin{displaymath}
  \begin{aligned}
    \langle pc_i^* \rangle_\sigma & \leftarrow
{\sf HSS.Output}(\M_{pc_i^*,\sigma},n_{\sf out}), \\ 
\langle v_i^* \rangle_\sigma & \leftarrow
{\sf HSS.Output}(\M_{v_i^*,\sigma},n_{\sf out}).
  \end{aligned}
\end{displaymath}

\State  {\bf Step 7}.
Apply the  permutation $\pi$ to the values in {\bf Step 6}: 
$\{ \langle pc_{\pi(i)}^* \rangle_{\sigma}, 
\langle v_{\pi(i)}^* \rangle_{\sigma}  \}_{i=1}^k$. Output 
\begin{equation}
  \label{eqn:outS}
  \left(\{ 
    \langle pc_{\pi(i)}^* \rangle_{\sigma},
    \langle v_{\pi(i)}^* \rangle_{\sigma}
    \}_{i=1}^k 
   \right).
  \end{equation}
\end{algorithmic}
\end{algorithm}

\noindent\textbf{Secure integer comparison.}
After the input preparation phase, each server $\calS_\sigma$ holds 
the HSS ciphertexts of the bits of the threshold values 
($\{\C_{y_{j,i}}\}_{i\in [t]}^{j\in [m]}$) 
and the HSS ciphertexts of the bits of the feature-selected features 
($\{\C_{(\mathcal{M}{\bf x})_{j,i}}\}_{i\in [t]}^{j\in [m]}$). 
The memory value $\M_{b_j,\sigma}$ of the comparison result 
at the $j$th decision node is computed using the \textsf{SIC} algorithm as follows:
\begin{displaymath}
  \M_{b_j,\sigma} \leftarrow {\sf SIC}(\sigma,\ek_\sigma, 
  \{\C_{(\mathcal{M}{\bf x})_{j,i}}\}_{i=1}^t\}, 
  \{\C_{y_{j,i}}\}_{i=1}^t\}).
\end{displaymath}   

\noindent\textbf{Secure result generation.}
Given the memory values $\{\M_{b_j,\sigma}\}_{j=1}^m$ 
at decision nodes $\{D_j\}_{j=1}^m$ 
and the HSS ciphertexts $\{\C_{v_i}\}_{i=1}^k$ 
at leaf nodes $\{L_j\}_{i=1}^k$, 
we design an \textsf{SRG} algorithm that allows
the servers to securely determine the final classification result
${\cal T}({\bf x})$. 
To this end, 
we adopt the {\em path cost} methodology from \cite{tai2017}, 
as also used in the secret-sharing based ODTE protocol \cite{zheng2022}, 
and tailor it for our HSS-based ODTE protocol.

The main idea of the path cost methodology is to assign
an {\em edge cost} to every outgoing edge of every decision node, 
compute the total edge cost (i.e., path cost) of a path for each leaf node, and
use the path cost to assess the likelihood of every leaf node containing ${\cal T}({\bf x})$.
In particular, 
 we assign the edge cost  $ec_{j,0}=1-b_j$ for edge $E_{j,0}$, 
 $ec_{j,1}=b_j$ for edge $E_{j,1}$ at each  decision node $D_j$, where $j\in [m]$. 
 Here, $E_{j,0}$ and $E_{j,1}$ are the left and right outgoing edge 
 of the decision node $D_j$, respectively. 
For every leaf node $L_i,i\in [k]$, the  path cost
$pc_i=\sum_{E_{j,\ell} \in P_i} ec_{j,\ell}$
is defined as the sum of the costs of all
edges on  $P_i$, the path from the root of $\cal T$ to the leaf $L_i$. 
The edge costs are chosen such that {\em a leaf node $L_i$ contains
the classification result  ${\cal T}({\bf x})$ if
and only if   $pc_i=0$.}
{\bf Algorithm \ref{alg:srg}} gives the details of our proposed \textsf{SRG} algorithm.

Similar to {\bf Algorithm \ref{alg:sic}}, {\bf Algorithm \ref{alg:srg}} starts its computation by 
generating a trivial memory value representing constant $1$. 
The subsequent operations are distributed between two non-communicating servers $\calS_\sigma,\sigma \in \bin$, 
which compute both edge costs and path costs.
Leveraging the additive homomorphism of HSS, each server computes: 
$\M_{ec_{j,0},\sigma}  \leftarrow {\sf HSS.Sub}(\M_{1,\sigma}, 
    \M_{b_j,\sigma})$, 
    and $\M_{ec_{j,1},\sigma}  \leftarrow \M_{b_j,\sigma}$.
This computation is valid since: $
  ec_{j,0} =  1- b_j = 
1- (\langle b_{j} \rangle_1  - \langle b_{j} \rangle_0) = 
(1- \langle b_{j} \rangle_1)  - (0-\langle b_{j} \rangle_0)$. 
The servers then convert the HSS ciphertexts $\C_{v_i}$ of each classification label $v_i$ into corresponding memory values $\M_{v_i,\sigma}$.
Note that direct return of path costs and classification labels is prohibited, as this would enable the client to recover
all path costs $pc_1,\ldots,pc_k$, 
not only identifying  the index 
$i\in[k]$ where $pc_i=0$ to determine
${\cal T}({\bf x})$, but also gaining additional information such as 
the nonzero path costs and the position $i$ of the leaf node that contains ${\cal T}({\bf x})$.
To hide the nonzero path costs, Step 5 of {\bf Algorithm \ref{alg:srg}} chooses
two random values  $r_{i,0}$ and $r_{i,1}$ for every leaf node
$L_i$, and masks the path cost shares and the leaf label shares,  such that
every non-zero path cost $ pc_{i'} (i'\neq i)$ becomes a random value $ pc_{i'}^*$ (after reconstruction) and
$pc_i$ becomes $ pc_i^*=0$.
{\em To conceal the position of the $i$ with $pc_i=0$},
Step 7 of {\bf Algorithm \ref{alg:srg}}
 requires each server to generated permuted values $\{ \langle pc_{\pi(i)}^* \rangle_{\sigma}  \}_{i=1}^k $ and
$\{ \langle v_{\pi(i)}^* \rangle_{\sigma}  \}_{i=1}^k $ to hide the indices of the path costs. 

\subsection{Result Reconstruction Phase}
Upon receiving the subtractive shares of randomized path costs and classification labels 
$\{\{\langle pc_{\pi(i)}^* \rangle_{\sigma}, 
  \langle v_{\pi(i)}^* \rangle_{\sigma}
  \}_{i=1}^k\}_{\sigma=0}^1 $
 from $\mathcal{S}_0$ and $\mathcal{S}_1$, 
 the client finds a 
 $\pi(i)\in[k]$ such that
$\langle pc_{\pi(i)}^* \rangle_{1} - \langle pc_{\pi(i)}^* \rangle_{0}=0$ 
and then obtains the classification result as  
${\cal T}({\bf x})= \langle v_{\pi(i)}^* \rangle_{1} - \langle v_{\pi(i)}^* \rangle_{0}$. 
Note that for any $i'\neq i$, we have    $\langle pc_{\pi(i')}^* \rangle_1 -
  \langle pc_{\pi(i')}^* \rangle_0 \neq 0$, and in this case, both  $\langle pc_{\pi(i')}^* \rangle_1 - $
  $ \langle pc_{\pi(i')}^* \rangle_0$ and $\langle v_{\pi(i')}^* \rangle_1 -
   \langle v_{\pi(i')}^* \rangle_0$ are random values 
   due to the techniques of using random masks and random permutations   in {\bf Algorithm \ref{alg:srg}}.

\subsection{Security Analysis}
\label{sec:semi1}
We treat the proposed protocol as an MPC protocol among the client, the model provider and two servers, and  follow the {\em ideal/real} world paradigm of \cite{zheng2023,zheng2022} to prove its security against a semi-honest client  and semi-honest  servers.
We begin with the following notations:

\noindent\textbf{Ideal Functionality.} Let $f^{\text{ODTE}}$ be a $\ppt$ functionality defined in the {\em ideal world} that captures the desired security properties of our ODTE protocol. 
  The functionality $f^{\text{ODTE}}$ takes as input a decision tree $\mathcal{T}$ from the model provider and a feature vector $\mathbf{x}$ from the client. It outputs the classification result $\mathcal{T}(\mathbf{x})$ to the client.

\noindent{\bf Real World Protocol.} Let $\Pi_{\text{ODTE}}$ denote our proposed protocol that realizes the functionality $f^{\text{ODTE}}$ in the {\em real world}. During an execution of protocol $\Pi_\text{ODTE}$, the output of server $\mathcal{S}_\sigma$ (where $\sigma \in \bin$) consists of additive shares as specified in Eq. \eqref{eqn:outS}: ${\rm output}_\sigma^{\Pi_\text{ODTE}}=(\{ \langle pc_{\pi(i)}^* \rangle_{\sigma}  \}_{i=1}^k ,
 \{ \langle v_{\pi(i)}^* \rangle_{\sigma}  \}_{i=1}^k )$. 
 The output of the client is the classification result ${\cal T}({\bf x})$. 
 The view of the client $\mathcal{U}$ during an execution of protocol $\Pi_{\text{ODTE}}$ is defined as:
$
{\mathbf{View}}_{\mathcal{U}}^{\Pi_{\text{ODTE}}} = (\pk,\mathbf{x}, \C_{\cal M}, {\rm output}_0^{\Pi_{\text{ODTE}}}, {\rm output}_1^{\Pi_{\text{ODTE}}})$, 
where $\pk$ is the public key of $\sf HSS$, 
$\mathbf{x}$ represents the client's feature vector and $\mathcal{C}_{\mathbf{M}}$ denotes the encrypted attribute mapping matrix received from the model provider.
The view of server $\mathcal{S}_\sigma$ during an execution of protocol $\Pi_{\text{ODTE}}$ comprises:
$
{\bf View}_{\calS_\sigma}^{\Pi_{\text{ODTE}}}=(\pk$, $\ek_\sigma,\{\C_{y_{j,i}}\}_{i\in [t]}^{j\in [m]}\cup
 \{\C_{{v}_i}\}_{i=1}^k\cup
 \{\C_{\mathcal{M}\mathbf{x}_{j,i}}\}_{i\in [t]}^{j\in [n]})$,
  where 
  $\pk$ is the public key of $\sf HSS$,
  $\ek_\sigma$ is the $\sigma$th evaluation key of $\sf HSS$, 
  $\C_{y_{j,i}}$ denotes the ciphertext of the $i$-th bit of the $j$-th element in vector $\mathbf{y}$,
  $\mathcal{C}_{{v}_{i}}$ represents the ciphertext of the $i$-th element in vector $\mathbf{v}$,
  $\C_{\mathcal{M}\mathbf{x}_{j,i}}$ corresponds to the ciphertext of the $(j,i)$-th element in matrix $\mathcal{M}\mathbf{x}$.

\begin{definition}[Semi-Honest Security]
The  protocol $\Pi_{\text{\em ODTE}}$ securely realizes the ideal functionality $f^{\text{\em ODTE}}$ {\it in the presence of semi-honest adversaries}
  if there exist $\ppt$ simulators 
  ${\bf Sim}_{\mathcal{U}}, {\bf Sim}_{\mathcal{S}_0}, {\bf Sim}_{\mathcal{S}_1}$, 
  such that $
  {\bf View}_{\mathcal{U}}^{\Pi_{\text{\em ODTE}}} \approx_c {\bf Sim}_{\mathcal{U}}({\bf x}, \mathcal{T}({\bf x}))$, and 
  ${\bf View}_{\mathcal{S}_\sigma}^{\Pi_{\text{\em ODTE}}} \approx_c{\bf Sim}_{\mathcal{S}_\sigma}, \sigma\in \bin$.
\end{definition}

\begin{theorem}
  Our protocol securely computes   the ideal functionality $f^{\text{\em ODTE}}$
  in the presence of semi-honest adversaries.
\end{theorem}

\begin{proof}
According to our security definition,
we need to construct the  simulators
${\bf Sim}_{\mathcal{U}}, {\bf Sim}_{\mathcal{S}_0}$ and 
${\bf Sim}_{\mathcal{S}_1}$.

\noindent \textbf{Simulator for the client - ${\bf Sim}_{\mathcal{U}}$.}
Given the feature vector  ${\bf x} \in \ZZ^n$ 
and the classification result
 ${\cal T}({\bf x})$, 
 ${\bf Sim}_{\mathcal{U}}$ firstly generates HSS keys: $(\pk,\ek_0,\ek_1)\gets {\sf HSS.Gen}(\secparam)$.
The distribution of $\C_{\cal M}$ can be simulated by simply choosing an $m\times n$ binary matrix ${\cal M}'$,
and encrypt each entry ${\cal M}'_{j,i}$ using ${\sf HSS.Input}(\pk,{\cal M}'_{j,i})$ to obtain $\C_{{\cal M}'}$. 
It is clearly that $\C_{\bf {\cal M}'}$ is computationally indistinguishable with $\C_{\cal M}$. 
Note that the messages ${\rm output}_0^{\Pi_{\text{ODTE}}}$ and ${\rm  output}_1^{\Pi_{\text{ODTE}}}$ received by the client 
 includes the secret shares of the permuted and randomized path costs and classification values.
The simulator can simulate these distributions as follows.
\begin{itemize}
\item
Choose an index $\eta \leftarrow  [k]$ uniformly at random and set  $pc_\eta = 0$. 
Choose   $pc_i \leftarrow \ZZ_N^*$ uniformly at random for all $i\in [k]\setminus \{\eta\}$.
\item
Set $v_\eta = \mathcal{T}({\bf x})$. Choose ${v}_i \leftarrow \ZZ_N^*$ uniformly at random for all $i\in[k]\setminus \{\eta\}$.
\item For every $i\in[k]$, choose the subtractive  shares
$\langle pc_i \rangle_0' \leftarrow \ZZ_N, \langle {v}_i \rangle'_0 \leftarrow \ZZ_N$ uniformly at random. Set 
$
\langle pc_i \rangle_1' = \langle pc_i \rangle_0' + pc_i \pmod{N},
\langle { v}_i \rangle_1' = \langle { v}_i \rangle_0' + { v}_i \pmod{N}
$.
 \item Set ${{\rm output'}_{\sigma}^{\Pi_{\text{ODTE}}}}=(\{ \langle pc_i \rangle_\sigma' \}_{i=1}^k, \{\langle {v}_i \rangle_\sigma'\}_{i=1}^k)$.
\end{itemize}

Due to the $\it security$ of the $\sf HSS$ during computation and
our randomization technique applied over the path costs and classification values respectively in our real protocol,
the position of the zero path cost and classification result is uniformly random in $\{1,\dots,k\}$,
and other non-zero path costs and corresponding classification labels are just random values.
So,
${\rm output}_\sigma^{\Pi_{\text{ODTE}}}$ is computationally indistinguishable with ${\rm output'}_{\sigma}^{\Pi_{\text{ODTE}}}$.
As such,
$
  {\bf View}_{\mathcal{U}}^{\Pi_{\text{ODTE}}} \approx_c {\bf Sim}_{\mathcal{U}}({\bf x}, \mathcal{T}({\bf x}))
$.

\noindent \textbf{Simulator for a server - ${\bf Sim}_{\mathcal{S}_\sigma}$.} 
${\bf Sim}_{\mathcal{S}_\sigma}$ firstly generates HSS keys: $(\pk,\ek_0,\ek_1)\gets {\sf HSS.Gen}(\secparam)$, 
and then encrypt $mt+k+nt$ dummy values to simulate 
$\{\C_{y_{j,i}}\}_{i\in [t]}^{j\in [m]}\cup
 \{\C_{{v}_i}\}_{i=1}^k\cup
 \{\C_{\mathcal{M}\mathbf{x}_{j,i}}\}_{i\in [t]}^{j\in [n]})$. 
Due to the \emph{security} of HSS, 
${\bf View}_{\mathcal{S}_\sigma}^{\Pi_{\text{ODTE}}} \approx_c {\bf Sim}_{\mathcal{S}_\sigma}$.
\end{proof}

\section{Our Maliciously Secure ODTE Protocol}
\label{sec:prot_malicious}
As discussed in Section \ref{sec:threat model},
the objective of ensuring security against malicious adversaries is to prevent two types of adversaries:
\begin{enumerate}
  \item A \textit{malicious server} (privacy attack) who attempts to deduce information about both the proprietary decision tree  $\mathcal{T}$  and
the sensitive feature vector $\bf x$; and
\item A \textit{malicious server} (verifiability attack) who attempts to manipulate the classification result  ${\cal T}({\bf x})$ 
and deceive the client.
\end{enumerate}

To our knowledge, current two-server ODTE protocols \cite{xu2023,liu2019,ma2021,zheng2023,zheng2022}, do not ensure security against two malicious adversaries. Only the PDTE protocol of \cite{ma2021}, which 
does not support outsourcability achieves security against malicious adversaries by replacing the underlying garbled circuit method with its maliciously secure version.

The challenge in designing a maliciously secure ODTE protocol lies in the necessity for multiple rounds of S2S communication. 
During the process of interactions, controlling malicious behaviors becomes extremely complex.
Fortunately, the servers in our ODTE protocol do not need to communicate with each other, thereby limiting malicious activities to local computations.

\subsection{Maliciously Secure HSS}
\label{sec:mhss}
A straightforward solution to achieve malicious security in our ODTE protocol is to replace the underlying HSS scheme with its maliciously secure version \cite{abram2022}. 
The core idea is to incorporate verifiable information within the memory values in a semi-honest HSS scheme. 
At the beginning of the protocol, the trusted party generates a Message Authentication Code (MAC) key $A$, and sends it to the client while distributing its linear share of $\langle A \rangle_0,\langle A \rangle_1$ to the servers, where $\langle A \rangle_1-\langle A \rangle_0=A$.
For any plaintext $x$, the corresponding memory value $\M_{x,\sigma}$ now comprises not only a linear share of $x$ but also a linear share of $A\cdot x$. 
And the algorithm $\sf HSS.Output(\M_{x,\sigma})$ is modified to yield $\langle x \rangle_\sigma$ along with $\langle A\cdot x \rangle_\sigma$. 
Upon receiving these linear shares $\langle x \rangle_0, \langle x \rangle_1, \langle A\cdot x \rangle_0, \langle A\cdot x \rangle_1$, 
the client can verify the correctness of the linear shares $\langle x \rangle_0, \langle x \rangle_1,$ by checking the equation: $
  \langle A\cdot x \rangle_1 - \langle A\cdot x \rangle_0 = A\cdot x$.

However, this approach has its drawbacks. 
First, for the verification process to be effective, the client must know the MAC key $A$. 
If the system serves multiple clients, the MAC key $A$ must be held among all participants, posing a risk to its confidentiality. 
There is no security if the server learns $A$.
Second, the computational overhead associated with a maliciously secure HSS is significantly higher. The memory value in a maliciously secure HSS is typically twice as extensive as that in a semi-honest HSS, at least doubling the computational complexity for each instruction in HSS.

\subsection{Our Maliciously Secure Method}
\label{sec:our_malicious}
To address these issues, we propose a similar yet effective method to achieve malicious security in our ODTE protocol. 

\noindent\textbf{MAC key preparation.} 
Instead of having a MAC key $A$ generated by a trusted party, 
we ask the client to generate a MAC key $A$ and encrypt it using $\sf HSS.Input$ to obtain $\C_A$ in the \emph{input preparation phase}. 
$\C_A$ is then distributed to two servers. 
As $A$ is only known to the client, and it is a one time key, the security of the MAC key is guaranteed. 

\noindent\textbf{Verifiable result generation.} 
In the \emph{sever-side computation phase}, we extend the \textsf{SRG} algorithm to \textsf{VRG} algorithm in {\bf Algorithm \ref{alg:vrg}}. 
Instead of doubling the length of the memory values, 
we perform a multiplication between the MAC key $A$ and the memory values of the classification labels to generate the verifiable classification labels  
in Step 6 in \textsf{VRG} algorithm. 
Note that there is no need to perform multiplication between the MAC key $A$ and the memory values of the path costs 
as the server does not know which path cost is zero. 
The probability of a malicious server manipulating the path costs is negligible.
Finally, the linear shares in Step 7 satisfy the equation:
$
    A\cdot (\langle v_{i}^* \rangle_1 - \langle v_{i}^* \rangle_0)  = 
    \langle w_{i}^* \rangle_1 - \langle w_{i}^* \rangle_0$.
Similar to the \textsf{SRG} algorithm, in Step 8, we perform a permutation $\pi$ on the linear shares.

\begin{algorithm}[t]
  \caption{$ (\{ \langle pc_{\pi(i)}^* \rangle_{\sigma},\langle v_{\pi(i)}^* \rangle_{\sigma},\langle w_{\pi(i)}^* \rangle_{\sigma}
  \}_{i=1}^k )  
  \leftarrow
{\sf VRG}(\sigma,$ $\ek_\sigma,\pi, 
\{ \M_{b_j,\sigma} \}_{j=1}^m,
\{\C_{v_i} \}_{i=1}^k,\C_A)$}\label{alg:vrg}
\begin{algorithmic} 
  \State {\bf Remark.} We use ``$\dots$'' to omit the unchanged inputs, outputs and computation steps in {\bf Algorithm \ref{alg:srg}}.
	\State  {\bf Input}.
    (1) $\dots$;
    (6) $\C_{A}$:
    the HSS ciphertexts of the MAC key $A$.
	\State {\bf Output}.
  (1) $\dots$;
  (3) $\{ \langle w_{\pi(i)}^* \rangle_{\sigma}  \}_{i=1}^k$:
  the subtractive shares of $k$ proofs of randomized classification labels.
  \State {\bf Step 1-5}. $\dots$

  \State {\bf Step 6}. Compute the proofs of the randomized classification label as: $
 \M_{w_i^*,\sigma} \leftarrow
{\sf HSS.Mul}(\C_A,\M_{v_i^*,\sigma})$.

\State {\bf Step 7}. Convert the memory values in {\bf Step 5} and {\bf Step 6} into linear shares: 
\begin{displaymath}
  \begin{aligned}
    \langle pc_i^* \rangle_\sigma & \leftarrow
{\sf HSS.Output}(\M_{pc_i^*,\sigma},n_{\sf out}), \\ 
\langle v_i^* \rangle_\sigma & \leftarrow
{\sf HSS.Output}(\M_{v_i^*,\sigma},n_{\sf out}), \\
\langle w_i^* \rangle_\sigma & \leftarrow
{\sf HSS.Output}(\M_{w_i^*,\sigma},n_{\sf out}).
  \end{aligned}
\end{displaymath}

\State  {\bf Step 8}.
Apply the  permutation $\pi$ to the values in {\bf Step 7}, output 
\begin{equation}
  \label{eqn:outvS}
  \left(\{ 
    \langle pc_{\pi(i)}^* \rangle_{\sigma},
    \langle v_{\pi(i)}^* \rangle_{\sigma},
    \langle w_{\pi(i)}^* \rangle_{\sigma}
    \}_{i=1}^k 
   \right).
  \end{equation}
\end{algorithmic}
\end{algorithm}

\noindent\textbf{Verifiable result reconstruction.}
Upon receiving the subtractive shares of randomized path costs, classification labels and proofs:
$
  \{\{\langle pc_{\pi(i)}^* \rangle_{\sigma}, 
  \langle v_{\pi(i)}^* \rangle_{\sigma},
  \langle w_{\pi(i)}^* \rangle_{\sigma}
  \}_{i=1}^k\}_{\sigma=0}^1$
 from $\mathcal{S}_0$ and $\mathcal{S}_1$, 
 the client finds a 
 $\pi(i)\in[k]$ such that
$\langle pc_{\pi(i)}^* \rangle_{1} - \langle pc_{\pi(i)}^* \rangle_{0}=0$ 
and then obtains the classification result as  
${\cal T}({\bf x})= \langle v_{\pi(i)}^* \rangle_{1} - \langle v_{\pi(i)}^* \rangle_{0}$.
To verify its correctness, the client can check the equation: $
A\cdot {\cal T}({\bf x}) = \langle w_{\pi(i)}^* \rangle_{1} - \langle w_{\pi(i)}^* \rangle_{0}$.

\subsection{Security Analysis}
\label{sec:anal}

We follow the standard ideal/real world paradigm to prove the security of our protocol against malicious servers. 
As discussed in Section \ref{sec:threat model}, we assume the servers may arbitrarily deviate from the protocol specification (e.g., by modifying the final shares), 
while the model provider and the client remain semi-honest (though the client verifies the result). 
We define an ideal functionality that captures the property of verifiability.

\noindent\textbf{Ideal Functionality.} 
Let $f^{\text{PVODTE}}$ be the functionality in the {\em ideal world} secure against malicious servers.
$f^{\text{PVODTE}}$ receives input $\mathcal{T}$ from the model provider and $\mathbf{x}$ from the client.
It computes the result $y = \mathcal{T}(\mathbf{x})$.
Before delivering $y$ to the client, $f^{\text{PVODTE}}$ allows the adversary (controlling the servers) to send a command. 
If the adversary sends ${\sf abort}$, $f^{\text{PVODTE}}$ sends $\bot$ to the client. 
Otherwise, it sends $y$ to the client.
Crucially, $f^{\text{PVODTE}}$ does not allow the adversary to modify $y$ into a valid $y' \neq y$.

The views in protocol $\Pi_{\text{PVODTE}}$ extend those in protocol $\Pi_{\text{ODTE}}$ with MAC-related messages: The view of the client $\mathcal{U}$ during an execution of protocol $\Pi_{\text{PVODTE}}$ is defined as:
${\mathbf{View}}_{\mathcal{U}}^{\Pi_{\text{PVODTE}}} = (A, 
\pk,\mathbf{x}, \C_{\cal M}, {\rm output}_0^{\Pi_{\text{PVODTE}}}, {\rm output}_1^{\Pi_{\text{PVODTE}}})$, 
where $A$ is the secret MAC key generated by the client, and 
${\rm output}_\sigma^{\Pi_{\text{PVODTE}}}=(\{ 
    \langle pc_{\pi(i)}^* \rangle_{\sigma},
    \langle v_{\pi(i)}^* \rangle_{\sigma},
    \langle w_{\pi(i)}^* \rangle_{\sigma}
    \}_{i=1}^k)$. 
    The view of server $\mathcal{S}_\sigma$ during an execution of protocol $\Pi_{\text{PVODTE}}$ comprises:
${\bf View}_{\calS_\sigma}^{\Pi_{\text{PVODTE}}}=\{{\bf View}_{\calS_\sigma}^{\Pi_{\text{ODTE}}},\C_A\}$,
  where $\C_A$ is the ciphertext of the MAC key $A$.

\begin{definition}[Malicious Security against Servers]
The protocol $\Pi_{\text{\em PVODTE}}$ securely realizes the ideal functionality $f^{\text{\em PVODTE}}$ in the presence of malicious servers if:
\begin{itemize}
    \item \textbf{Privacy:} 
    There exist $\ppt$ simulators
     ${\bf Sim}_{\mathcal{U}}$, ${\bf Sim}_{\mathcal{S}_0}$, ${\bf Sim}_{\mathcal{S}_1}$, 
  such that $
  {\bf View}_{\mathcal{U}}^{\Pi_{\text{\em PVODTE}}} \approx_c {\bf Sim}_{\mathcal{U}}({\bf x}, \mathcal{T}({\bf x}))$, and 
  ${\bf View}_{\mathcal{S}_\sigma}^{\Pi_{\text{\em PVODTE}}} \approx_c{\bf Sim}_{\mathcal{S}_\sigma}, \sigma\in \bin$.
    \item \textbf{Verifiability:} For any malicious adversary $\mathcal{A}$ controlling the servers, the probability that the client outputs an incorrect result $y' \neq \mathcal{T}(\mathbf{x})$ and $y' \neq \bot$ is negligible.
\end{itemize}
\end{definition}

\begin{theorem}
  The protocol $\Pi_{\text{\em PVODTE}}$ securely computes the ideal functionality $f^{\text{\em PVODTE}}$
  in the presence of malicious adversaries.
\end{theorem}
\begin{proof} 

The privacy proof mirrors Section~\ref{sec:semi1} with added MAC handling, so we focus on verifiability. 
Given the symmetry between servers in $\Pi_{\text{PVODTE}}$, it suffices to assume $\adv$ corrupted $\calS_0$. 
Let the true result be $v_\eta={\cal T}({\bf x})$ at index $\eta$ with $pc^*_\eta=0$. 
To force an incorrect output ${\cal T}'({\bf x}) \neq {\cal T}({\bf x})$, $\adv$ must attempt the following attacks:

\noindent\textbf{Attack 1: Forge a zero path cost at index $j \neq \eta$.}
The adversary modifies $\langle pc_{j}^{*} \rangle_{0}$ to $\langle pc_{j}^{*} \rangle_{0}'$ aiming for $\langle pc_{j}^{*} \rangle_{1} - \langle pc_{j}^{*} \rangle_{0}' = 0$.
However, for $j \neq \eta$, the value $pc_j \in \ZZ_N^*$ is uniformly random. 
Supported by the semantic security of HSS and the security  of PRFs, the remote share $\langle pc_j \rangle_1$ remains indistinguishable from a uniform random value over $\ZZ_N$  to $\adv$.
Therefore, $\adv$ must guess $\langle pc_j \rangle_0' = \langle pc_j \rangle_1$ with probability at most $1/N$ per index. 
By a union bound over all $k-1$ indices:
$
\Pr[\text{Attack 1 succeeds}]=\frac{k-1}{N}(1-\frac{1}{N})^{k-1} < \frac{k-1}{N} = \negl$ considering $N \gg k$.

\noindent\textbf{Attack 2: Forge the classification value at the correct index.}
We prove by contradiction. Suppose $\adv$ successfully forges a classification result ${\cal T}'({\bf x}) \neq {\cal T}({\bf x})$ that passes the verification. 
For the correct classification result at index $\eta$, the MAC verification equation is:
$
A \cdot v_\eta = w_\eta$,
where $v_\eta = \langle v_\eta \rangle_1 - \langle v_\eta \rangle_0 = {\cal T}({\bf x})$ and $w_\eta = \langle w_\eta \rangle_1 - \langle w_\eta \rangle_0$.
If $\adv$ successfully forges a result by modifying its shares to $\langle v_\eta \rangle_0'$ and $\langle w_\eta \rangle_0'$, the forged values must satisfy:
$A \cdot v'_\eta = w'_\eta$,
where $v'_\eta = \langle v_\eta \rangle_1 - \langle v_\eta \rangle_0' = {\cal T}'({\bf x}) \neq {\cal T}({\bf x})$ and $w'_\eta = \langle w_\eta \rangle_1 - \langle w_\eta \rangle_0'$.
From these two equations, we have:
$
A \cdot (v'_\eta - v_\eta) = w'_\eta - w_\eta$.
Since $v'_\eta = {\cal T}'({\bf x}) \neq {\cal T}({\bf x}) = v_\eta$, we have $v'_\eta - v_\eta \neq 0$. Therefore, $\adv$ can compute the MAC key as:
$
A = \frac{w'_\eta - w_\eta}{v'_\eta - v_\eta} = \frac{\langle w_\eta \rangle_0 - \langle w_\eta \rangle_0'}{\langle v_\eta \rangle_0 - \langle v_\eta \rangle_0'}.
$
Note that all terms in the right-hand side are known to $\adv$.
However, this contradicts the security of $\mathsf{HSS}$. By the security of $\mathsf{HSS}$, the ciphertext $\C_A$ should not reveal any information about $A$ to $\adv$. 
Specifically, $A$ should remain computationally indistinguishable from a uniformly random value in $\ZZ_N$ from $\adv$'s view. If $\adv$ could compute $A$ as shown above with non-negligible probability, this would contradict the security of $\mathsf{HSS}$.
Therefore:
$\Pr[\text{Attack 2 succeeds}] = \negl$.

\end{proof}

\begin{algorithm}[t]
  \caption{$ (\{ \langle pc_{\pi(i)}^* \rangle_{\sigma},\langle v_{\pi(i)}^* \rangle_{\sigma},\langle w_{\pi(i)}^* \rangle_{\sigma}
  \}_{i=1}^k )  
  \leftarrow
{\sf VRG_{GBDT}}(\sigma,$ $\ek_\sigma,\pi, 
\{ \M_{b_j,\sigma} \}_{j=1}^m,
\{\C_{v_i} \}_{i=1}^k,\C_A,C_\eta,r_j,r_j')$}\label{alg:vrg_gbdt}
\begin{algorithmic} 
  \State {\bf Remark.} We use ``$\dots$'' to omit the unchanged inputs, outputs and computation steps in {\bf Algorithm \ref{alg:vrg}}.
	\State  {\bf Input}.
    (1) $\dots$;
    (7) $\C_{\eta}$:
    the HSS ciphertexts of the learning rate $\eta$.
    (8) $r_j$:
    A random mask for classification labels in the tree $\mathcal{T}_j$.
    (9) $r_j'$:
    A random mask for proofs in the decision tree $\mathcal{T}_j$.
	\State {\bf Output}.
  (1) $\dots$; 
  (2) $\{ \langle \mu_{\pi(i)}^* \rangle_{\sigma}  \}_{i=1}^k$: the weighted-outputs of the randomized classification labels.
  (3) $\{ \langle \tau_{\pi(i)}^* \rangle_{\sigma}  \}_{i=1}^k$: the weighted-proofs of the randomized classification labels.
  \State {\bf Step 1-6}. $\dots$

    \State {\bf Step 7}. Compute the weighted-output of the randomized classification label as: $
 \M_{\mu_i^*,\sigma} \leftarrow
{\sf HSS.Mul}(\C_\eta,\M_{v_i^*,\sigma})$.

  \State {\bf Step 8}. Compute the weight-proofs of the randomized classification label as: $
 \M_{\tau_i^*,\sigma} \leftarrow
{\sf HSS.Mul}(\C_\eta,\M_{w_i^*,\sigma})$.

\State {\bf Step 9}. Convert the memory values in {\bf Step 7} and {\bf Step 8} into linear shares: 
\begin{displaymath}
  \begin{aligned}
    \langle pc_i^* \rangle_\sigma & \leftarrow
{\sf HSS.Output}(\M_{pc_i^*,\sigma},n_{\sf out}), \\ 
\langle \mu_i^* \rangle_\sigma & \leftarrow
{\sf HSS.Output}(\M_{\mu_i^*,\sigma},n_{\sf out}), \\
\langle \tau_i^* \rangle_\sigma & \leftarrow
{\sf HSS.Output}(\M_{\tau_i^*,\sigma},n_{\sf out}).
  \end{aligned}
\end{displaymath}
\State {\bf Step 10}. Add random masks $r_j$, $r_j'$ to the linear shares of the weighted outputs and weighted proofs: $
      \langle \mu_i^* \rangle_\sigma \leftarrow 
  \langle \mu_i^* \rangle_\sigma + r_j$, $ 
      \langle \tau_i^* \rangle_\sigma \leftarrow
  \langle \tau_i^* \rangle_\sigma + r_j'$.

\State  {\bf Step 11}.
Apply the  permutation $\pi$ to the values in {\bf Step 10}, output 
\begin{equation}
  \label{eqn:outvGBDT}
  \left(\{ 
    \langle pc_{\pi(i)}^* \rangle_{\sigma},
    \langle \mu_{\pi(i)}^* \rangle_{\sigma},
    \langle \tau_{\pi(i)}^* \rangle_{\sigma},
    \}_{i=1}^k 
   \right).
  \end{equation}
\end{algorithmic}
\end{algorithm}

\section{Further Investigations} 
In this section, we extend $\sf PVODTE$ to support two practical inference scenarios~\cite{wu2015}: Gradient Boosted Decision Trees and decision trees using categorical features.

\subsection{Extension to Gradient Boosted Decision Trees}
\label{sec:Forest}
A gradient boosted decision tree (GBDT)~\cite{FRIEDMAN2002367} is an ensemble learning model comprising sequentially trained decision trees, where the final prediction is a weighted sum of outputs from all constituent trees (each scaled by a small learning rate $\eta$). Formally, given a GBDT model $\mathcal{T}_{\rm GBDT}$ with $s$ trees $\{\mathcal{T}_1, \ldots, \mathcal{T}_s\}$, the prediction for an input $\mathbf{x}$ is: $
  \mathcal{T}_{\rm GBDT}(\mathbf{x}) = T_0 + \eta \sum_{j=1}^s \mathcal{T}_j(\mathbf{x})$, 
where $T_0$ is the initial constant bias (estimated from training data). 
For simplicity, we treat $T_0$ as the output of a dummy tree $\mathcal{T}_0$.
To extend $\sf PVODTE$ to GBDTs, we employ an additive masking technique to preserve the privacy of individual tree predictions while ensuring the final prediction remains unbiased. 
The core idea is to inject random masks into each tree's output such that: (1) individual tree predictions are hidden from the client, and (2) the sum of masks cancels out (i.e., $\sum_{j=0}^s r_j=\sum_{j=0}^s r_j' = 0$, where $r_j$, $r_j'$ are the masks for tree $\mathcal{T}_j$).

The extended protocol for GBDTs modifies the original $\sf PVODTE$ workflow as follows:

\noindent{\bf Input preparation phase.} The model provider encrypts the initial bias $T_0$ and learning rate $\eta$ using the $\sf HSS.Input$ algorithm, generating ciphertexts $\C_{T_0}$ and $\C_\eta$. These, along with the encrypted GBDT model, are sent to servers $\mathcal{S}_0$ and $\mathcal{S}_1$. 

\noindent{\bf Server-side computation phase.} Instead of executing the single-tree $\sf VRG$ algorithm, 
each server $\mathcal{S}_\sigma$ runs a modified $\sf VRG_{GBDT}$ algorithm (detailed in Algorithm \ref{alg:vrg_gbdt}) with inputs $\C_{T_0}$, $\C_\eta$, and two random masks $r_j$, $r_j'$ 
(satisfying $\sum_{j=0}^s r_j = \sum_{j=0}^s r_j' =0$). 
After executing $\sf VRG_{GBDT}$ for each tree $\mathcal{T}_j$, each server $\mathcal{S}_\sigma$ additionally masks and sends the secret share of $T_0$ to the client. The share is derived as: 
$ \M_{T_0,\sigma} \leftarrow \mathsf{HSS.ConvertInput}(\ek_\sigma,\C_{T_0})$, $ \langle T_0 \rangle_\sigma \leftarrow {\sf HSS.Output}(\M_{T_0,\sigma},n_{\sf out})$, and then masked as $\langle T_0 \rangle_\sigma \leftarrow \langle T_0 \rangle_\sigma + r_0$. 
For correctness verification of $T_0$, the servers compute: 
$\M_{T_0',\sigma} \leftarrow \mathsf{HSS.Mul}(\C_A, \M_{T_0,\sigma})$, $\langle T_0' \rangle_\sigma \leftarrow {\sf HSS.Output}(\M_{T_0',\sigma},n_{\sf out})$, and $\langle T_0' \rangle_\sigma \leftarrow \langle T_0' \rangle_\sigma + r_0'$.

\noindent{\bf Result reconstruction phase.} 
  The client receives subtractive secret shares of randomized path costs, weighted classification labels, and weighted proofs from both servers for each tree $\mathcal{T}_j$: $
  \{\{\langle pc_{\pi(i)}^* \rangle_{\sigma}, 
  \langle \mu_{\pi(i)}^* \rangle_{\sigma},
  \langle \tau_{\pi(i)}^* \rangle_{\sigma}
  \}_{i=1}^k\}_{\sigma=0}^1$
 from $\mathcal{S}_0$ and $\mathcal{S}_1$.  
 To identify the correct output for each decision tree $\mathcal{T}_j$, the client finds a
 $\pi(i)\in[k]$ such that
$\langle pc_{\pi(i)}^* \rangle_{1} - \langle pc_{\pi(i)}^* \rangle_{0}=0$, 
and then obtains the weighted-output 
$\eta \cdot {\cal T}_j({\bf x})$, 
and its proof $A \cdot \eta \cdot {\cal T}_j({\bf x})$ as: $
  \eta \cdot {\cal T}_j({\bf x}) = \langle \mu_{\pi(i)}^* \rangle_{1} - \langle \mu_{\pi(i)}^* \rangle_{0}$ and
$A \cdot \eta \cdot {\cal T}_j({\bf x}) = \langle \tau_{\pi(i)}^* \rangle_{1} - \langle \tau_{\pi(i)}^* \rangle_{0}$. 
The final GBDT prediction is then computed by aggregating the reconstructed tree outputs and the initial bias:
$
  {\cal T}_{\rm GBDT}({\bf x}) =\langle T_0 \rangle_1 - \langle T_0 \rangle_0  + \cdot \sum_{j=1}^s \eta \cdot {\cal T}_j({\bf x})$.
To verify correctness, the client checks that the aggregated proof matches the product of $A$ and the final prediction: $
  A\cdot {\cal T}_{\rm GBDT}({\bf x}) =\langle T_0' \rangle_1 - \langle T_0' \rangle_0  + \cdot \sum_{j=1}^s A\cdot \eta \cdot {\cal T}_j({\bf x})$.

The correctness of the extended protocol follows directly from the correctness of the $\sf PVODTE$ and the linearity of the HSS operations used for aggregation.
The privacy of $T_0$ and $\eta$ is preserved by the \emph{security} of the HSS scheme, ensuring they remain private from the servers. 
For the client, only the final GBDT prediction and its proof are revealed—aligning with $\sf PVODTE$'s security requirements. 
Individual tree predictions are obscured by random masks $r_j$, $r_j'$, which cancel out in the aggregated sum, preventing the client from recovering intermediate tree outputs.

\begin{algorithm}[t]
  \caption{$\M_{c_{t},\sigma}\leftarrow {\sf SEQ}(\sigma, \ek_\sigma, \{\C_{\alpha_i}\}_{i=1}^t, \{\C_{\beta_i}\}_{i=1}^t)$}  \label{alg:seq}
\begin{algorithmic}
\State  {\bf Input}. (1)
    $\sigma\in \{0,1\}$: the index of a  server ${\cal S}_\sigma$;
  (2) $\ek_\sigma$: the private evaluation key of HSS;
  (3) $\C_{\alpha_i}$ (resp. $\C_{\beta_i}$): a ciphertext  of the bit  $\alpha_i$ (resp. $\beta_i$)
  under  ${\sf HSS}$, where $1\leq i\leq t$.
\State {\bf Output}.  $\M_{c_{t},\sigma}$: a memory value of the bit $c_{t}$.
\State {\bf Step 1}. $\calS_\sigma$ generates a trivial memory value ${\sf M}_{1,\sigma}$;
\State {\bf Step 2}. ${\sf M}_{\alpha_1,\sigma} \leftarrow {\sf HSS.ConvertInput}(\ek_\sigma,\C_{\alpha_1})$;
\State {\bf Step 3}. ${\sf M}_{\beta_1,\sigma} \leftarrow {\sf HSS.ConvertInput}(\ek_\sigma,\C_{\beta_1})$;
\State  {\bf Step 4}. ${\sf M}_{c_1,\sigma} \leftarrow {\sf HSS.Mul}(\C_{\beta_1},\M_{\alpha_1,\sigma})$;
\State  {\bf Step 5}. ${\sf M}_{c_1,\sigma} \leftarrow {\sf HSS.Add}({\sf M}_{c_1,\sigma},{\sf M}_{c_1,\sigma})$;
\State  {\bf Step 6}. ${\sf M}_{c_1,\sigma} \leftarrow {\sf HSS.Sub}({\sf M}_{c_1,\sigma},{\sf M}_{\alpha_1,\sigma})$;
\State  {\bf Step 7}. ${\sf M}_{c_1,\sigma} \leftarrow {\sf HSS.Sub}({\sf M}_{c_1,\sigma},{\sf M}_{\beta_1,\sigma})$;
\State  {\bf Step 8}. ${\sf M}_{c_1,\sigma} \leftarrow {\sf HSS.Add}({\sf M}_{c_1,\sigma},{\sf M}_{1,\sigma})$;
\State  {\bf Step 9}.
{\bf for $i\in \{1,\ldots,t-1\}$ do}
\State \hspace{15mm}
  $\C_{temp} \leftarrow {\sf HSS.Add}(\C_{\alpha_{i+1}},\C_{\beta_{i+1}} )$;
  \State \hspace{15mm}
  ${\sf M}_{c_{i+1},\sigma} \leftarrow {\sf HSS.Mul}(\C_{temp},{\sf M}_{c_i,\sigma})$;
  \State \hspace{15mm}    ${\sf M}_{c_{i+1},\sigma} \leftarrow {\sf HSS.Sub}({\sf M}_{c_i,\sigma}, \M_{c_{i+1},\sigma})$;
\State \hspace{15mm}    ${\sf M}_{2c_i,\sigma} \leftarrow {\sf HSS.Add}({\sf M}_{c_i,\sigma}, {\sf M}_{c_i,\sigma})$;
\State \hspace{15mm}    ${\sf M}_{temp,\sigma} \leftarrow {\sf HSS.Mul}(\C_{\alpha_{i+1}},{\sf M}_{2c_i,\sigma})$;
\State \hspace{15mm}    ${\sf M}_{temp,\sigma} \leftarrow {\sf HSS.Mul}(\C_{\beta_{i+1}},{\sf M}_{temp,\sigma})$;
\State \hspace{15mm}   ${\sf M}_{c_{i+1},\sigma} \leftarrow {\sf HSS.Add}({\sf M}_{c_{i+1},\sigma}, {\sf M}_{temp,\sigma})$;
\State {\bf Step 10}. Output ${\sf M}_{c_{t},\sigma}$.
\end{algorithmic}
\end{algorithm}

\subsection{Extension to Categorical Features}
In practice, feature vectors $\bf x$ often include categorical variables (e.g., ``gender'' or ``product type'') instead of numerical ones. 
For such variables, the natural branching operation at a decision node (i.e., the Boolean test function) is \emph{set membership testing} (e.g., ``is $x$ in the set $S$?'')~\cite{wu2015}. 
 
Consider a categorical variable $x$ in $\bf x$ corresponding to a non-leaf node $D_j$, where $x$ takes values from a finite set $S = \{s_1, \ldots, s_\ell\}$. 
The Boolean test function at $D_j$ checks whether $x \in S$ holds. 
Translating this into our $\sf PVODTE$ requires a construction that: (i) Takes as input the HSS ciphertext $\C_x$ of $x$ and the HSS ciphertexts $\{\C_{s_1}, \ldots, \C_{s_\ell}\}$ of elements in $S$.  
(ii) Outputs a memory value of $1$ if $x = s_j$ for any $s_j \in S$, and $0$ otherwise.  
 
Note that we have already designed a ``greater than'' testing algorithm (Algorithm \ref{alg:sic}) for bitwise integer comparison (Section \ref{sec:sic}). 
This bit-level approach extends naturally to set membership testing: we compare $x$ sequentially with each element in $S$, and the final result is the sum of these pairwise equality tests.  
 
\noindent{\bf Secure equality test algorithm.}
To implement set membership testing, we first design a \emph{secure equality test algorithm} ($\sf SEQ$) to compare two $t$-bit variables $\alpha = (\alpha_1, \ldots, \alpha_t)$ and $\beta = (\beta_1, \ldots, \beta_t)$ (where $\alpha_1$, $\beta_1$ are least significant bits (LSBs)). 
Recall that for any two bits $\alpha_i, \beta_i \in \{0,1\}$, the equality test can be expressed via arithmetic over $\ZZ$ as:  
\begin{equation}
  (\alpha_i = \beta_i) := 1 - \alpha_i - \beta_i + 2\alpha_i\beta_i,
  \label{eq:eq}
\end{equation}  
where $(\alpha_i = \beta_i):=1$ if $\alpha_i = \beta_i$ and $0$ otherwise.  

Using an \emph{LSB approach}, we iteratively construct the equality result for the full $t$-bit variables. 
Let $c_0 = 1$, and let $c_i$, $1\leq i \leq t$ denote the equality result for the first $i$ bits ($\alpha_i \ldots \alpha_1 = \beta_i \ldots \beta_1$). 
The recursive relation for $c_{i+1}$ is:  
\begin{equation}
  c_{i+1} := (\alpha_{i+1} = \beta_{i+1}) \wedge c_i, \quad 0 \leq i < t,
  \label{eq:lsbbeq}
\end{equation}  
where $c_t = 1$ if and only if $\alpha = \beta$. 
By substituting Eq.~\eqref{eq:gt} and Eq.~\eqref{eq:eq} into Eq.~\eqref{eq:lsbbeq}, 
we obtain:  
\begin{displaymath}
  c_{i+1} := c_i - c_i(\alpha_{i+1} + \beta_{i+1}) + 2c_i\alpha_{i+1}\beta_{i+1}, \quad 1 \leq i < t.
\end{displaymath}  
Similar to the $\sf SIC$ algorithm (Algorithm~\ref{alg:sic}), we design $\sf SEQ$ (Algorithm~\ref{alg:seq}) to implement this equation. 
The $\sf SEQ$ algorithm also keeps the privacy of the input $x$, the set $S$ and the equality test result from the servers.

\noindent\textbf{Set membership testing via $\sf SEQ$.}
For a categorical variable $x$ and set $S = \{s_1, \ldots, s_\ell\}$, each server $\mathcal{S}_\sigma$: 
(1) Run $\sf SEQ$ for each $s_j \in S$ to compute a memory value $\M_{c_j,\sigma}$ ($c_j=1$ if $x = s_j$, $0$ otherwise).  
(2) Sum the memory values of $c_j$ to obtain the set membership result:  
$\M_{c_S,\sigma} = \sum_{j=1}^\ell \M_{c_j,\sigma}$, 
where $c_S = 1$ if $x \in S$ and $0$ otherwise. 

\begin{table}[t]
  \renewcommand{\arraystretch}{1.1}
  \centering
  \caption{
    Tree Parameters in Our Experiments.
  }\label{tab:tree}
  \begin{tabular}{c||c|c|c}
    \hline
    Decision Tree & \#Features & Tree height & \#Nodes \\
    \hline
    \textsf{heart-disease} & 13 & 3 & 15  \\
    \textsf{breast-cancer} & 9 & 8 & 511 \\
    \textsf{housing} & 13 & 13 & 16383 \\
    \textsf{spambase} & 57 & 17 & 262143 \\ 
    \textsf{MNIST} & 784 & 20 & 2.1E+06 \\
    \hline
  \end{tabular}
  \vspace{2mm}
  \centering
  \caption{Comparison of running time of $\sf PVODTE_{MH}$ and $\sf PVODTE_{MS}$} \label{tab:vsmhss}
  \renewcommand{\arraystretch}{1.2}
  \resizebox{1\columnwidth}{!}{
    \begin{threeparttable}
    \begin{tabular}{c|cc|cc}
    \toprule
    \multirow{2}{*}{\textbf{ Decision tree}}
  & \multicolumn{2}{c|}{$\sf SIC$ (s)} & \multicolumn{2}{c}{$\sf VRG$ (s)} \\ 
   & $\sf PVODTE_{MH}$ & $\sf PVODTE_{MS}$ & $\sf PVODTE_{MH}$ & $\sf PVODTE_{MS}$ \\
  \hline
  \textsf{heart-disease} & 4.53 & 2.18  & 0.14 & 0.07 \\
  \textsf{breast-cancer} & 4.41 & 2.23  & 2.82 & 1.30 \\
  \textsf{housing} & 4.10 & 2.25 & 87.16 & 42.53 \\
  \textsf{spambase} & 4.34 & 2.12 & 1352.13 & 670.26 \\ 
  \textsf{MNIST} & 4.21  & 2.29 & 11139.41 & 5373.95\\ 
  \bottomrule
\end{tabular}
\end{threeparttable}
  }
  \vspace{-4mm}
\end{table}

\section{Experiments} 
\label{sec:expe}

In this section, we provide detailed evaluation results for $\sf PVODTE$.

\noindent\textbf{Dataset and model.} 
These datasets have been used in prior work \cite{wu2015,zheng2023,zheng2022,bai2023} to evaluate ODTE protocols. 
For each dataset, we train a decision tree using the scikit-learn toolkit \cite{pedregosa2011}, which allows setting the maximum depth of the resulting tree. Following standard practice \cite{kiss2019,wu2015,zheng2023,zheng2022}, we pad the tree with dummy nodes to reach the target depth.
We train six decision-tree models, as shown in TABLE \ref{tab:tree}. Among them, \textsf{heart-disease} and \textsf{breast-cancer} represent small trees, while \textsf{housing} and \textsf{spambase} are medium-sized trees with a moderate number of features and classes. \textsf{MNIST} is a deep tree with a high-dimensional feature vector; we include it primarily as a benchmark, though decision trees are not commonly used for classification tasks in practice.

    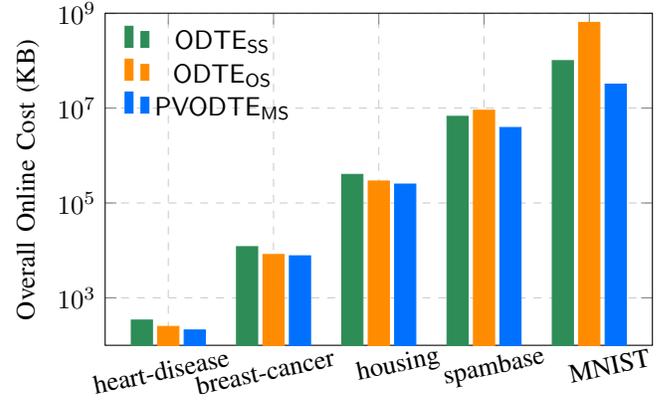
\begin{figure}[t] 
    \centering \begin{tikzpicture} 
      \begin{axis}[ 
        width=\columnwidth, height=6cm, ybar, bar width=8pt, enlarge x limits=0.15, 
        ylabel={Overall Online Cost (KB)}, 
        ymode=log, 
        log basis y=10, ymin=1e2, ymax=1e9, 
        symbolic x coords={ heart, breast, housing, spam, MNIST }, 
        xtick=data, 
        xticklabels={ heart-disease, breast-cancer, housing, spambase, MNIST }, 
        x tick label style={rotate=12, anchor=east,xshift=25pt,yshift=-9pt}, 
        tick align=inside,
        legend style={ at={(0.02,0.98)}, 
        anchor=north west, 
        draw=none, fill=none }, 
        grid=both, major grid style={dashed, gray!40}, ] \addplot[draw=seagreen,fill=seagreen] coordinates { (heart,3.4e2) (breast,1.2e4) (housing,4.0e5) (spam,6.7e6) (MNIST,1.0e8) }; 
        \addplot[draw=darkorange,fill=darkorange] coordinates { (heart,2.5e2) (breast,8.2e3) (housing,2.9e5) (spam,9.0e6) (MNIST,6.4e8) }; 
        \addplot[draw=brandeisblue,fill=brandeisblue] coordinates { (heart,2.1e2) (breast,7.7e3) (housing,2.5e5) (spam,3.9e6) (MNIST,3.2e7) }; 
        \legend{ $\sf ODTE_{SS}$, $\sf ODTE_{OS}$, $\sf PVODTE_{MS}$ } 
      \end{axis} 
    \end{tikzpicture} 
    \caption{Comparison of overall online communication overhead of $\sf ODTE_{SS}$ \cite{zheng2022}, $\sf ODTE_{OS}$ \cite{zheng2023} and our $\sf PVODTE_{MS}$. The y-axis is in logarithmic scale.} \label{fig:overall-online} 
    \vspace{-4mm} 
  \end{figure}
  
\noindent\textbf{Implementation setup.} 
We implement our semi-honestly secure ODTE protocol and maliciously secure ODTE protocol, 
denoted as $\sf ODTE_{SH}$ and  $\sf PVODTE_{MS}$, respectively. 
Both protocols use a semi-honest ElGamal-based HSS scheme from \cite{abram2022}, which satisfies the definitions of the HSS scheme in Section \ref{sec:HSS_Paillier}.
To evaluate its performance, we compare our $\sf PVODTE_{MS}$ with three alternative protocols:
\begin{enumerate}
  \item A maliciously secure ODTE protocol, denoted as $\sf PVODTE_{MH}$, constructed by combining our semi-honest ODTE protocol with the maliciously secure HSS scheme from \cite{abram2022}. 
  \item The state-of-the-art two-server semi-honest ODTE protocol, denoted as $\sf ODTE_{SS}$, proposed by Zheng et al. \cite{zheng2022}, which is based on a secret sharing scheme.
  \item The optimized version of the protocol of \cite{zheng2022}, denoted as $\sf ODTE_{OS}$, proposed by Zheng et al. \cite{zheng2023}, which is designed to reduce the S2S rounds in \cite{zheng2022}.
\end{enumerate} 
All protocols are implemented in C++ using the NTL library \cite{shoup2011ntl}.
We set a feature or a threshold value to be a random $t=10$ bit integer, 
which is a common choice in existing experiments \cite{lu2018,xu2023}. 
Our protocol's output modulus and the message spaces of $\sf ODTE_{SS}$ and $\sf ODTE_{OS}$ are all set to 128 bits, while the RSA modulus in our protocol is set to $3072$ bits. 
We implement the OT protocol employed  in \cite{zheng2023} as a classic OT based on RSA encryption (the RSA modulus is same as ours), since the authors do not provide the details of the OT protocol they employ. 

To test the performance of our protocols in a genuine heterogeneous network setting, we conduct benchmarks on a workstation equipped with Intel(R) Xeon(R) Gold 6250 CPU @ 3.90GHz running Ubuntu 22.04 LTS with 64GB memory. 
We employ the Linux tc tool to simulate local-area network (LAN,
RTT: 0.1 ms, 1 Gbps), and wide-area network (WAN, RTT: 160 ms, 135 Mbps). 
This process allows us to learn the real-world performance of our protocol during model inference, such as actual online communication overhead and running time. 
Our source code is available at \url{https://github.com/neuroney/PODT}.

\subsection{Lightweight MAC for Malicious Security: Comparison with $\sf PVODTE_{MH}$} 
To demonstrate the lightweight efficiency of our approach, 
we compare ${\sf PVODTE_{MS}}$ against a baseline protocol ${\sf PVODTE_{MH}}$ built directly on maliciously secure HSS. As detailed in Section \ref{sec:mhss}, major distinctions in MAC key preparation, memory value size, and server-side computation translate precisely into the performance differences observed below.

\noindent\textbf{Communication overhead.}
The primary divergence stems from MAC key scaling constraints. While ${\sf PVODTE_{MH}}$ relies on a trusted setup issuing lightweight linear shares (256 bytes) \cite{abram2022}, ${\sf PVODTE_{MS}}$ acts over a client-provided HSS ciphertext causing a specific overhead difference (3072 bytes per ciphertext versus 256 for shares). Crucially, aside from this initialization step, all other communication burdens—inputs and outputs—remain nearly identical between the two frameworks.

\noindent\textbf{Running time.}
The memory value size in $\sf PVODTE_{MH}$ is twice that of $\sf PVODTE_{MS}$, leading to reduced efficiency when the cloud server executes HSS instructions. 
Server-side computation comprises two key steps: executing the $\sf SIC$ algorithm for $m=2^h-1$ non-leaf (decision) nodes and the $\sf VRG$ algorithm for $k=2^h$ leaf nodes. 
Comparison results are summarized in TABLE \ref{tab:vsmhss}. 
Compared with $\sf PVODTE_{MH}$, $\sf PVODTE_{MS}$ achieves a roughly $50\%$ reduction in running time for both the $\sf SIC$ and $\sf VRG$ algorithms.

\begin{table*}[t!]
  \renewcommand{\arraystretch}{1.2}

\caption{Client-side latency breakdown of $\sf ODTE_{SS}$, $\sf ODTE_{OS}$, $\sf ODTE_{SH}$ and $\sf PVODTE_{MS}$ 
  over different Datasets} 
   \label{tab:clientlatency}
   \centering
            \resizebox{1\textwidth}{!}{
     \begin{threeparttable}
     \begin{tabular}{ccccccccccc}
     \toprule
           \multirow{2}{*}{\textbf{Protocols}} & \multicolumn{2}{c}{\textsf{heart-disease}} & \multicolumn{2}{c}{\textsf{breast-cancer}} & \multicolumn{2}{c}{\textsf{housing}} &\multicolumn{2}{c}{\textsf{spambase}} & \multicolumn{2}{c}{\textsf{MINIST}} \\ 
           \cmidrule(lr){2-3}\cmidrule(lr){4-5}\cmidrule(lr){6-7}\cmidrule(lr){8-9}\cmidrule(lr){10-11}
          & C$\rightharpoonup  $S  & C$\leftharpoondown $S &
           C$\rightharpoonup  $S  & C$\leftharpoondown $S &
            C$\rightharpoonup  $S  & C$\leftharpoondown $S &
        C$\rightharpoonup  $S  & C$\leftharpoondown $S &
           C$\rightharpoonup  $S & C$\leftharpoondown $S \\
          \midrule $\sf ODTE_{SS}$ \cite{zheng2022} &  
          \textbf{5.47E--05} & {3.49E--05}  & 
          \textbf{3.79E--05} & {1.12E--03} & 
          \textbf{5.47E--05} & {3.57E--02} & 
          \textbf{2.40E--04} & {5.71E--01} & 
          \textbf{3.30E--03} & {4.57E+00} \\ 
           $\sf ODTE_{OS}$ \cite{zheng2023} &  
          \textbf{5.47E--05} & \textbf{2.18E--06}  & 
          \textbf{3.79E--05} & \textbf{2.18E--06} & 
          \textbf{5.47E--05} & \textbf{2.18E--06} & 
          \textbf{2.40E--04} & \textbf{2.18E--06} & 
          \textbf{3.30E--03} & \textbf{2.18E--06} \\ 
            $\sf ODTE_{SH}$ & 
           {5.66E--02} & 3.49E--05  & 
           {2.06E+00} & {1.12E--03} & 
           {6.62E+01} & {3.57E--02} & 
           {1.06E+03} & {5.71E--01} &
           {8.47E+03} & {4.57E+00} \\    
           $\sf PVODTE_{MS}$ & 
           {5.66E--02} & {5.23E--05}  &
           {2.06E+00} & {1.67E--03} &
           {6.62E+01} & {5.36E--02} & 
           {1.06E+03} & {8.57E--01} & 
           {8.47E+03} & {6.86E+00} \\            
         \bottomrule
     \end{tabular}
            
     Remark. 
     All latency values are reported in seconds (s). 
     We measure ``C$\rightharpoonup  $S'' and ``C$\leftharpoondown $S'' using the global median upload/download speeds (58 Mbps for upload, 112 Mbps for download) as reported by Ookla \cite{ookla2025}. 
  \end{threeparttable}
            }
\vspace{1mm}
  \caption{Online running time of $\sf ODTE_{SS}$, $\sf ODTE_{OS}$, $\sf ODTE_{SH}$ and $\sf PVODTE_{MS}$ 
  over different Datasets and Networks} 
   \label{tab:serverlatency}
   \centering
     \begin{threeparttable}
     \begin{tabular}{ccccccccccc}
     \toprule
           \multirow{2}{*}{\textbf{Protocols}} & \multicolumn{2}{c}{\textsf{heart-disease}} & \multicolumn{2}{c}{\textsf{breast-cancer}} & \multicolumn{2}{c}{\textsf{housing}} &\multicolumn{2}{c}{\textsf{spambase}} & \multicolumn{2}{c}{\textsf{MINIST}} \\ 
           \cmidrule(lr){2-3}\cmidrule(lr){4-5}\cmidrule(lr){6-7}\cmidrule(lr){8-9}\cmidrule(lr){10-11}
           & LAN.T & WAN.T& LAN.T & WAN.T& LAN.T & WAN.T& LAN.T & WAN.T& LAN.T & WAN.T  \\
          \midrule $\sf ODTE_{SS}$ \cite{zheng2022} &  
          \textbf{2.61E--02}  & {4.08E+01} & 
          {4.17E--01} & {6.50E+02} & 
          {1.34E+01} & {2.08E+04} & 
          {2.16E+02} & {3.33E+05} & 
          {2.0E+03} & {2.66E+06} \\ 
           $\sf ODTE_{OS}$ \cite{zheng2023} &  
           {8.30E--02}  & {3.44E+00} & 
           \textbf{1.02E--01} & {6.66E+01} & 
           \textbf{1.80E+00} & {2.54E+03} & 
           \textbf{3.12E+01} & {4.59E+04} & 
           \textbf{2.71E+02} & {3.98E+05} \\ 
            $\sf ODTE_{SH}$ & 
           2.25E+00  & \textbf{2.25E+00} & 
           {3.63E+01} & \textbf{3.63E+01} & 
           {1.16E+03} & \textbf{1.16E+03} & 
           {1.86E+04} & \textbf{1.86E+04} & 
           {1.48E+05} & \textbf{1.48E+05} \\    
           $\sf PVODTE_{MS}$ & 
           2.33E+00  & {2.33E+00} & 
           {3.76E+01} & {3.76E+01} & 
           {1.20E+03} & {1.20E+03} & 
           {1.92E+04} & {1.92E+04} & 
           {1.54E+05} & {1.54E+05} \\     
           \textsf{IF}\textsubscript{1} & 
           /  & 17.51 & 
           / & 17.29 & 
           / & 17.33 & 
           / & 17.34 & 
           / & 17.27 \\       
           \textsf{IF}\textsubscript{2} & 
           /  & 1.48 & 
           / & 1.77 & 
           / & 2.12 & 
           / & 2.39 & 
           / & 2.58 \\         
         \bottomrule

     \end{tabular}
     Remark. {LAN.T} and {WAN.T} represent online running time under LAN (RTT: 0.1ms), and online running time under WAN (RTT: 160ms), respectively. All running times are reported in seconds (s). The minimum value in each column is bolded. \textsf{IF}\textsubscript{1} represents the improvement factor of $\sf PVODTE_{MS}$ over $\sf ODTE_{SS}$ (i.e. $\sf ODTE_{SS}$/$\sf PVODTE_{MS}$) and \textsf{IF}\textsubscript{2} represents the improvement factor of $\sf PVODTE_{MS}$ over $\sf ODTE_{OS}$, (i.e. $\sf ODTE_{OS}$/$\sf PVODTE_{MS}$). 
  \end{threeparttable}
\vspace{-4mm} 
\end{table*}

 \begin{table}
      \caption{Overall latency of $\sf ODTE_{SS}$, $\sf ODTE_{OS}$, $\sf ODTE_{SH}$ and $\sf PVODTE_{MS}$ 
  over different Datasets in WANs} 
   \label{tab:overalllatency}
   \centering
            \resizebox{1\columnwidth}{!}{
     \begin{threeparttable}
     \begin{tabular}{cccccc}
     \toprule
           \textbf{Protocols} & \textsf{heart-disease} & \textsf{breast-cancer} & \textsf{housing} & \textsf{spambase} & \textsf{MINIST} \\ 
          \midrule $\sf ODTE_{SS}$ [1] &  {4.08E+01} & {6.50E+02} & 
           {2.08E+04} & {3.33E+05} & {2.66E+06} \\ 
           $\sf ODTE_{OS}$ [2] &  
           {3.44E+00} & 
            {6.66E+01} & 
           {2.54E+03} & 
           {4.59E+04} & 
           {3.98E+05} \\ 
            $\sf ODTE_{SH}$ & 
           \textbf{2.31E+00} 
           & \textbf{3.84E+01}  
           & \textbf{1.23E+03}  
           & \textbf{1.97E+04}  
           & \textbf{1.56E+05} \\    
           $\sf PVODTE_{MS}$ & 
           {2.39E+00} & {3.97E+01} & {1.27E+03} & {2.03E+04} & 
           {1.62E+05} \\     
           \textsf{IF}\textsubscript{1} & 
            17.07 & 
           16.37 & 
           16.38 & 
           16.40 & 
           16.42 \\       
           \textsf{IF}\textsubscript{2} & 
           1.44 & 
           1.68 & 
           2.00 & 
           2.26 & 
           2.47 \\         
         \bottomrule

     \end{tabular}
            
     Remark. The overall latency is computed as the sum of client-server communication latency (``C$\rightharpoonup$S'' and ``C$\leftharpoonup$S'' in TABLE \ref{tab:clientlatency}) and the server-side online computation time (``WAN.T'' in TABLE \ref{tab:serverlatency}) in WAN settings.
     All running times are reported in seconds (s). The minimum value in each column is bolded. \textsf{IF}\textsubscript{1} represents the improvement factor of $\sf PVODTE_{MS}$ over $\sf ODTE_{SS}$ (i.e. $\sf ODTE_{SS}$/$\sf PVODTE_{MS}$) and \textsf{IF}\textsubscript{2} represents the improvement factor of $\sf PVODTE_{MS}$ over $\sf ODTE_{OS}$, (i.e. $\sf ODTE_{OS}$/$\sf PVODTE_{MS}$). 
  \end{threeparttable}
            }
            \vspace{-7mm}
 \end{table}

\begin{table*}[t!]
  \renewcommand{\arraystretch}{1.2}
  \caption{Online Communication overhead of $\sf ODTE_{SS}$, $\sf ODTE_{OS}$, $\sf ODTE_{SH}$ and $\sf PVODTE_{MS}$ 
  over different Datasets} 
   \label{tab:online communication}
   \centering
     \begin{threeparttable}
     \begin{tabular}{ccccccccccc}
     \toprule
           \multirow{2}{*}{\textbf{Protocols}} & \multicolumn{2}{c}{\textsf{heart-disease}} & \multicolumn{2}{c}{\textsf{breast-cancer}} & \multicolumn{2}{c}{\textsf{housing}} &\multicolumn{2}{c}{\textsf{spambase}} & \multicolumn{2}{c}{\textsf{MINIST}} \\ 
           \cmidrule(lr){2-3}\cmidrule(lr){4-5}\cmidrule(lr){6-7}\cmidrule(lr){8-9}\cmidrule(lr){10-11}
          & \#Client. & \#Server. & \#Client. & \#Server. & \#Client. & \#Server. & \#Client. & \#Server. & \#Client. & \#Server.   \\
          \midrule $\sf ODTE_{SS}$ \cite{zheng2022} &  
          9.06E--01  & {3.40E+02} & 
          {1.63E+01} & {1.23E+04} & 
          {5.12E+02} & {3.96E+05} & 
          {8.19E+03} & {6.71E+06} & 
          {6.56E+04} & {1.01E+08} \\ 
           $\sf ODTE_{OS}$ \cite{zheng2023} &  
           \textbf{2.19E--01}  & {2.45E+02} & 
           \textbf{1.55E--01} & {8.24E+03} & 
           \textbf{2.19E--01} & {2.92E+05} & 
           \textbf{9.06E--01} & {9.03E+06} & 
           \textbf{1.23E+01} & {6.44E+08} \\ 
           $\sf ODTE_{SH}$ &  
           2.10E+02  & \textbf{0} & 
           {7.66E+03} & \textbf{0} & 
           {2.46E+05} & \textbf{0} & 
           {3.94E+06} & \textbf{0} & 
           {3.15E+07} & \textbf{0} \\ 
           $\sf PVODTE_{MS}$ & 
           2.13E+02  & \textbf{0} & 
           {7.67E+03} & \textbf{0} & 
           {2.46E+05} & \textbf{0} & 
           {3.94E+06} & \textbf{0} & 
           {3.15E+07} & \textbf{0} \\       
         \bottomrule
     \end{tabular}
     Remark. \#Client. and \#Server. represent the number of bits required for online computation respectively. All communications are 
     reported in KB. The minimum value in each column is bolded.
  \end{threeparttable}
  \vspace{-4mm}
  \end{table*}

\subsection{Comparison with $\sf ODTE_{SS}$ and $\sf ODTE_{OS}$}
Our protocol is designed to enhance real-time efficiency 
during the critical online phase of decision tree evaluation. 
By optimizing this phase, 
we address the most urgent performance bottlenecks in practical applications. 
The key metrics are the overall time taken for one inference, which includes both 
client-server communication latency and server-side computation time, 
and the online communication overhead, which includes both client-side and server-side communication. 
In this section, we present experimental results focusing on these key metrics.

\noindent\textbf{Client-side communication latency.}
We evaluate client-side communication latency for ODTE protocols by estimating the time for data upload (sending private inputs to the servers) and data download (receiving the inference result). As shown in TABLE \ref{tab:clientlatency}, $\sf PVODTE_{MS}$ does not have an advantage in upload time compared with $\sf ODTE_{SS}$ and $\sf ODTE_{OS}$. Both $\sf PVODTE_{MS}$ and $\sf ODTE_{SH}$ require the client to upload larger ciphertexts than $\sf ODTE_{SS}$ and $\sf ODTE_{OS}$. For download time, $\sf ODTE_{SS}$, $\sf ODTE_{SH}$, and $\sf PVODTE_{MS}$ have similar latency since they all return all path costs and leaf-node values to the client, whereas $\sf ODTE_{OS}$ returns only the class label and therefore achieves the lowest latency. Although $\sf PVODTE_{MS}$ incurs higher client-side communication latency than $\sf ODTE_{SS}$ and $\sf ODTE_{OS}$, this cost is small relative to server-side computation, as shown in TABLE \ref{tab:serverlatency} and TABLE \ref{tab:overalllatency}.

\noindent\textbf{Server-side running time.} 
The server side dominates the end-to-end latency in ODTE protocols. TABLE \ref{tab:serverlatency} reports the online running time and contrasts our protocol $\sf PVODTE_{MS}$ with $\sf ODTE_{SS}$ \cite{zheng2022} and $\sf ODTE_{OS}$ \cite{zheng2023}. While $\sf ODTE_{SS}$ and $\sf ODTE_{OS}$ achieve lower running time in LAN settings, $\sf PVODTE_{MS}$ substantially outperforms both in WAN environments. The primary reason is that $\sf ODTE_{SS}$ and $\sf ODTE_{OS}$ rely on frequent joint computations at each decision and leaf node, which amplifies synchronization and network latency. In contrast, $\sf PVODTE_{MS}$ employs efficient local evaluation, which is robust to network latency. Concretely, in WAN settings our protocol is approximately $17\times$ faster than $\sf ODTE_{SS}$ and $1.5\mbox{-}2.6\times$ faster than $\sf ODTE_{OS}$.

\noindent\textbf{Overall latency.} 
As summarized in TABLE \ref{tab:overalllatency}, $\sf PVODTE_{MS}$ achieves significant overall latency reductions over $\sf ODTE_{SS}$ and $\sf ODTE_{OS}$ in WAN settings. On average, it is about $16.4\times$ faster than $\sf ODTE_{SS}$ and $1.5\mbox{-}2.5\times$ faster than $\sf ODTE_{OS}$. These gains primarily stem from the server-side design choices that mitigate synchronization in high-latency networks.

\noindent\textbf{Online communication overhead.} 
We estimate the total online communication for our maliciously secure ODTE protocol by accounting for client-server traffic and server-to-server traffic. Client-side communication includes ciphertexts uploaded by the client and results returned. Server-side communication captures messages exchanged between servers during computation. TABLE \ref{tab:online communication} compares our protocol with the two-server ODTE baselines $\sf ODTE_{SS}$ and $\sf ODTE_{OS}$. $\sf ODTE_{SS}$ and $\sf ODTE_{OS}$ incur lower client-side online communication than ours because secret-sharing ciphertexts are smaller than the HSS ciphertexts used by $\sf PVODTE_{MS}$. However, our protocol requires zero server-to-server communication due to the absence of S2S rounds, while $\sf ODTE_{OS}$ requires up to 6.4E+08 KB. Consequently, as shown in Fig.~\ref{fig:overall-online}, the total online communication of $\sf PVODTE_{MS}$ is $1.6\mbox{-}3.1\times$ lower than $\sf ODTE_{SS}$ and $1.1\mbox{-}20\times$ lower than $\sf ODTE_{OS}$.

As the decision-tree depth exceeds $h=13$, the overall communication advantage of $\sf PVODTE_{MS}$ becomes more pronounced relative to $\sf ODTE_{SS}$ and $\sf ODTE_{OS}$. The improvement arises because our protocol eliminates S2S rounds, which grow exponentially with depth in $\sf ODTE_{SS}$ and $\sf ODTE_{OS}$. Therefore, for deeper trees, our protocol scales more favorably and offers superior communication efficiency in large-scale deployments.

\subsection{Overhead Analysis of Malicious Security in $\sf PVODTE$} 
To quantify the overhead introduced by malicious security, we first compare the performance of our semi-honest ($\sf ODTE_{SH}$) and maliciously secure ($\sf PVODTE_{MS}$) protocols.

\noindent\textbf{Communication overhead.} 
The client sends an HSS ciphertext of a MAC key $A$, which is around $3$ KB to both servers. Compared with the encrypted feature vector (around 210 KB) sent to the servers, the cost for verification is small. 
Each server sends $3k$ linear shares ($48k$ bytes) to the client for verification, while $2k$ linear shares ($32k$ bytes) in $\sf ODTE_{SH}$. The additional cost for verification is $16k$ bytes.
As shown in TABLE~\ref{tab:online communication},
the overhead taken by malicious security in $\sf PVODTE_{MS}$ is negligible compared with the overall communication overhead of $\sf ODTE_{SH}$.

\noindent\textbf{Computation overhead.}
In our experiments, the $\sf HSS.Input$ and $\sf HSS.Mul$ each take around 81 ms. 
The client requires an additional 81 ms to encrypt the MAC key $A$ via $\sf HSS.Input$, plus an extremely lightweight verification step (0.001 ms) that multiplies two plaintext values to verify MAC tags.
Each server performs $k$ instances of $\sf HSS.Mul$ (81 ms per instance) to compute MAC tags at leaf nodes, totaling $81k$ ms of extra computation. As shown in TABLE ~\ref{tab:serverlatency}, the server-side running time of $\sf PVODTE_{MS}$ remains roughly on par with $\sf ODTE_{SH}$. This is because the dominant overhead of our protocol, comparing feature values at non-leaf nodes is identical across both security models.

\subsection{WAN sensitivity}
To further analyze the performance improvements, we conducted experiments across WAN settings (RTT 40 to 200 ms).
Figure \ref{fig:WAN} confirms that $\sf PVODTE_{MS}$ surpasses $\sf ODTE_{SS}$ and $\sf ODTE_{OS}$ once RTT $>$ 80ms, distinguishing it as superior for latency-critical applications.

It's worth mentioning that by implementing our protocol with high-performance programming techniques and deploying it on more advanced servers will surely further reduce the evaluation time for deeper decision trees significantly. 
This is because the servers in our protocol execute all computations locally,
thereby circumventing the adverse impacts of online communication. 
In contrast, the protocols in \cite{zheng2022,zheng2023} face execution time bottlenecks due to bandwidth constraints, network delays, and the need to handle extensive S2S messages.

\begin{figure*}[t]
  \centering
  
  {\footnotesize
  \tikz[baseline=-0.5ex]{\draw[mark=pentagon*,color=seagreen,mark options={scale=1.0},thick] plot coordinates {(0,0) (0.3,0)};} $\sf ODTE_{SS}$ [1] \quad
  \tikz[baseline=-0.5ex]{\draw[mark=diamond*,color=darkorange,mark options={scale=1.0},thick] plot coordinates {(0,0) (0.3,0)};} $\sf ODTE_{OS}$ [2] \quad
  \tikz[baseline=-0.5ex]{\draw[mark=square*,color=brandeisblue,mark options={scale=1.0},thick] plot coordinates {(0,0) (0.3,0)};} $\sf PVODTE_{MS}$
  }
  \vspace{-2mm}
  
  \subfloat[\textsf{heart-disease}]{
    \label{fig:WAN11}
    \centering
  \begin{tikzpicture}[scale=0.52]
    \begin{axis}[%
    symbolic x coords={40,80,120,160,200},
    scaled ticks=false, 
    height=4.8cm,
    width=7.5cm,
    grid=major,
    grid style=dashed,
    ymode=log,
    log basis y={10},
    ymin=0.7,
    ymax=2E+02,
    ylabel={Running time (s)},
    xlabel={WAN RTT (ms)},
    xlabel style={font=\large},
    ylabel style={font=\large},
    tick label style={font=\large},
    ]
    
    \addplot[mark=pentagon*,color=seagreen,mark options={scale=1.2},thick] 
    coordinates {(40,10.20064939)(80,20.40064939)(120,30.60064939)(160,40.80064939)(200,51.00064939)};

    \addplot[mark=diamond*,color=darkorange,mark options={scale=1.2},thick] 
    coordinates {(40,0.92085288)(80,1.76085288)(120,2.60085288)(160,3.44085288)(200,4.28085288)};
    
    \addplot[mark=square*,color=brandeisblue,mark options={scale=1.2},thick] 
    coordinates {(40,2.3332)(80,2.3332)(120,2.3332)(160,2.3332)(200,2.3332)};
    \end{axis}
    \end{tikzpicture}%
      }%
\hspace{-1mm}%
\subfloat[\textsf{breast-cancer}]{
  \centering
  \label{fig:WAN2}
  \begin{tikzpicture}[scale=0.52]
    \begin{axis}[%
    symbolic x coords={40,80,120,160,200},
    scaled ticks=false, 
    height=4.8cm,
    width=7.5cm,
    grid=major,
    grid style=dashed,
    ymode=log,
    log basis y={10},
    ymin=5E+00,
    ymax=2E+03,
    ylabel={},
    xlabel={WAN RTT (ms)},
    xlabel style={font=\large},
    ylabel style={font=\large},
    tick label style={font=\large},
    ]
    \addplot[mark=pentagon*,color=seagreen,mark options={scale=1.2},thick] 
    coordinates {(40,162.6106578)(80,325.2106578)(120,487.8106578)(160,650.4106578)(200,813.0106578)};

    \addplot[mark=diamond*,color=darkorange,mark options={scale=1.2},thick] 
    coordinates {(40,16.69990128)(80,33.33990128)(120,49.97990128)(160,66.61990128)(200,83.25990128)};
    
    \addplot[mark=square*,color=brandeisblue,mark options={scale=1.2},thick] 
    coordinates {(40,37.56)(80,37.56)(120,37.56)(160,37.56)(200,37.56)};
    \end{axis}
    \end{tikzpicture}%
}%
\hspace{-1mm}%
\subfloat[\textsf{housing}]{
  \centering
  \label{fig:WAN3}
  \begin{tikzpicture}[scale=0.52]
    \begin{axis}[%
    symbolic x coords={40,80,120,160,200},
    scaled ticks=false, 
    height=4.8cm,
    width=7.5cm,
    grid=major,
    grid style=dashed,
    ymode=log,
    log basis y={10},
    ymin=3E+2,
    ymax=2E+05,
    ylabel={},
    xlabel={WAN RTT (ms)},
    xlabel style={font=\large},
    ylabel style={font=\large},
    tick label style={font=\large},
    ]
    
    \addplot[mark=pentagon*,color=seagreen,mark options={scale=1.2},thick] 
    coordinates {(40,5202.313852)(80,10404.27385)(120,15606.23385)(160,20808.19385)(200,26010.15385)};

    \addplot[mark=diamond*,color=darkorange,mark options={scale=1.2},thick] 
    coordinates {(40,635.0936074)(80,1269.973607)(120,1904.853607)(160,2539.733607)(200,3174.613607)};
    
    \addplot[mark=square*,color=brandeisblue,mark options={scale=1.2},thick] 
    coordinates {(40,1201.9696)(80,1201.9696)(120,1201.9696)(160,1201.9696)(200,1201.9696)};
    \end{axis}
    \end{tikzpicture}%
}%
\hspace{-1mm}%
\subfloat[\textsf{spambase}]{
  \centering
 \label{fig:WAN4}
  \begin{tikzpicture}[scale=0.52]
    \begin{axis}[%
    symbolic x coords={40,80,120,160,200},
    scaled ticks=false, 
    height=4.8cm,
    width=7.5cm,
    grid=major,
    grid style=dashed,
    ymode=log,
    log basis y={10},
    ymin=5E+3,
    ymax=2E+06,
    xlabel={WAN RTT (ms)},
    xlabel style={font=\large},
    ylabel style={font=\large},
    tick label style={font=\large},
    ]
    
    \addplot[mark=pentagon*,color=seagreen,mark options={scale=1.2},thick] 
    coordinates {(40,83238.67085)(80,166469.4308)(120,249700.1908)(160,332930.9508)(200,416161.7108)};

    \addplot[mark=diamond*,color=darkorange,mark options={scale=1.2},thick] 
    coordinates {(40,11471.30413)(80,24940.10413)(120,34408.90413)(160,45877.70413)(200,57346.50413)};
    
    \addplot[mark=square*,color=brandeisblue,mark options={scale=1.2},thick] 
    coordinates {(40,19231.5376)(80,19231.5376)(120,19231.5376)(160,19231.5376)(200,19231.5376)};
    \end{axis}
    \end{tikzpicture}%
}%
\hspace{-1mm}%
  \subfloat[\textsf{MNIST}]{
  \centering
  \label{fig:WAN5}
  \begin{tikzpicture}[scale=0.52]
    \begin{axis}[%
    symbolic x coords={40,80,120,160,200},
    scaled ticks=false, 
    height=4.8cm,
    width=7.5cm,
    grid=major,
    grid style=dashed,
    ymode=log,
    log basis y={10},
    ymin=5E+4,
    ymax=2E+07,
    ylabel={},
    xlabel={WAN RTT (ms)},
    xlabel style={font=\large},
    ylabel style={font=\large},
    tick label style={font=\large},
    ]
    
    \addplot[mark=pentagon*,color=seagreen,mark options={scale=1.2},thick] 
    coordinates {(40,666206.3895)(80,1332052.189)(120,1997897.989)(160,2663743.789)(200,3329589.589)};

    \addplot[mark=diamond*,color=darkorange,mark options={scale=1.2},thick] 
    coordinates {(40,99636.92228)(80,199251.6423)(120,298866.3623)(160,398481.0823)(200,498095.8023)};
    
    \addplot[mark=square*,color=brandeisblue,mark options={scale=1.2},thick] 
    coordinates {(40,153852.312)(80,153852.312)(120,153852.312)(160,153852.312)(200,153852.312)};
    \end{axis}
    \end{tikzpicture}%
}
    
       \caption{Online running time in different WAN settings for different databases. $y$-axis is in the logarithm scale.}
       \label{fig:WAN}
       \vspace{-4mm}
  \end{figure*}
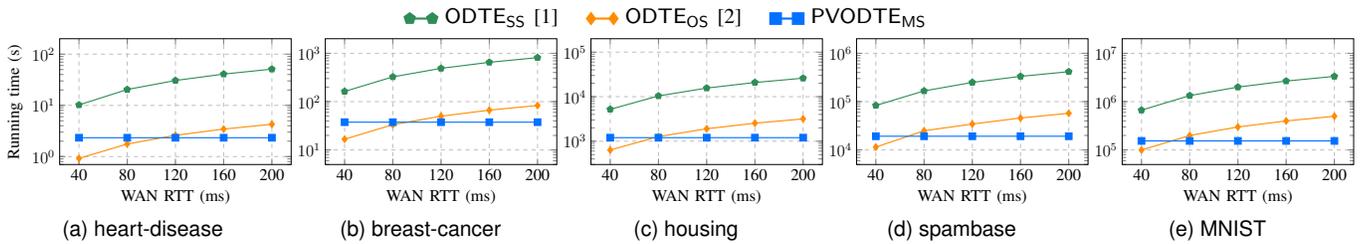

\begin{table}[t]
  \centering
  \renewcommand{\arraystretch}{1.2}
  \caption{
    One-time offline download cost for encrypted matrix $\C_{\cal M}$.
  }\label{tab:cm}
  \resizebox{1\columnwidth}{!}{
  \begin{threeparttable}
  \begin{tabular}{c||c|c|c}
    \hline
    Decision Tree & \#$\mathsf{C}_M$ (KB) & Construction time (s) & Download time (s) \\
    \hline
    \textsf{heart-disease} & 2.73E+02 & 4.55E+00 & 1.90E--02  \\
    \textsf{breast-cancer} & 6.89E+03 & 1.15E+02 & 4.80E--01  \\
    \textsf{housing} & 3.19E+05 & 5.32E+03 & 2.23E+01  \\
    \textsf{spambase} & 2.24E+07 & 3.74E+05 & 1.56E+03  \\ 
    \textsf{MNIST}* & 2.47E+09 & 4.11E+07 & 1.72E+05  \\
    \hline
  \end{tabular}
    $^*$\textsf{MNIST} is included primarily for benchmarking and is not a typical decision-tree dataset in practice.
    Download time is estimated using the global median download speed (112 Mbps) reported by Ookla~\cite{ookla2025}.
  \end{threeparttable}
  }
  \vspace{-4mm}
\end{table}

\subsection{Offline Cost Analysis}
\label{sec:offline cost}
In addition to the online phase, our protocol includes an offline phase in which the client performs a one-time download of the encrypted matrix $\C_{\cal M}$ (Algorithm~\ref{alg:sfs}). Although this cost is reusable and amortizable across multiple inferences, we report it to provide a complete evaluation. The size of $\C_{\cal M}$ is $3mn$ KB, where $m$ is the number of non-leaf nodes and $n$ is the number of features. As shown in TABLE~\ref{tab:cm}, the offline communication cost is practical for all datasets except \textsf{MNIST} at depth 20, which we include solely as a stress-test benchmark rather than a representative real-world use case. 

An important observation is that $\C_{\cal M}$ encodes only the tree topology (i.e., which feature is tested at each decision node), not the actual threshold values or classification labels. Consequently, if the model is updated by adjusting thresholds or labels while maintaining the tree structure, $\C_{\cal M}$ requires no redistribution. Even when the tree structure itself is modified, only the affected rows of $\C_{\cal M}$ (one per modified non-leaf node) need to be recomputed and transmitted, significantly reducing the update cost compared to full redistribution.

\section{Conclusion}
\label{sec:conc}

In this paper, we propose a novel two-server ODTE protocol 
$\sf PVODTE$,  
that eliminates the requirement for S2S rounds  
and achieves security against malicious servers simultaneously.
To our knowledge, this is the first ODTE protocol that achieve all these properties. 
By enabling each server to independently process secret-shared data, 
our protocol avoids communication and computation overhead typically associated with bandwidth constraints and network delays in WAN environments. 
Moreover, despite its enhanced security and efficiency benefits, our proof-of-concept implementation shows performance that remains comparable to existing protocols.

\bibliographystyle{IEEEtran}

\renewcommand{\appendixname}{Supplementary Material}
\appendix
\label{sec:eliminate-offline}
Recall that in the \textsf{SFS} algorithm, the client downloads the encrypted mapping matrix $\C_{\mathcal{M}}$ once and reuses it to locally compute $\C_{\mathcal{M}\mathbf{x}}$ via the additive homomorphism of \textsf{HSS}. The key reason the client must hold $\C_{\mathcal{M}}$ locally is that the computation exploits the client's private input $\mathbf{x}$: since $x_{s,i}$ is a known plaintext bit, multiplying $\C_{\mathcal{M}_{j,s}}$ by $x_{s,i}$ reduces to a conditional selection—either including or excluding $\C_{\mathcal{M}_{j,s}}$ from the summation without requiring any ciphertext-ciphertext multiplication. If, however, we wish to delegate this computation entirely to the servers so that the client need only upload her encrypted feature vector without downloading $\C_{\mathcal{M}}$, the servers must evaluate the product of two ciphertexts, namely $\C_{\mathcal{M}_{j,s}}$ and $\C_{x_{s,i}}$.

This operation is not supported by the group-based \textsf{HSS} schemes~\cite{abram2022,orlandi2021,boyle2016}, whose underlying encryption scheme is linearly homomorphic over a group and therefore admits no ciphertext-ciphertext multiplication. 
In contrast, an \textsf{HSS} instantiation based on the RLWE assumption~\cite{boyle2019} employs a \emph{somewhat} homomorphic encryption (SHE) scheme as its underlying primitive, which supports one level of ciphertext-ciphertext multiplication. Under such an instantiation, the client encrypts her feature vector $\mathbf{x}$ bit by bit as $\{\C_{x_{s,i}}\}$ via \textsf{HSS.Input} for $s \in [n], i \in [t]$ and uploads the result to the servers; the model provider sends $\C_{\mathcal{M}}$ directly to the servers rather than making it publicly downloadable. 
Each server then computes:
\begin{displaymath}
\C_{(\mathcal{M}\mathbf{x})_{j,i}} \leftarrow \sum_{s=1}^{n} {\sf CMult}(\C_{\mathcal{M}_{j,s}},\C_{x_{s,i}}),
\end{displaymath}
for every $j\in [m], i\in [t]$, where ${\sf CMult}$ denotes the ciphertext-ciphertext multiplication supported by the underlying SHE scheme, and the summation denotes repeated application of $\sf Add$. The servers thereby obtain $\C_{\mathcal{M}\mathbf{x}}$ and proceed with the remainder of the protocol as before.

The trade-off, however, is non-trivial. In RLWE-based \textsf{SHE}, each ciphertext-ciphertext multiplication increases the multiplicative depth by one, which causes the ciphertext size to at least double. Consequently, the server-side storage and computation costs at least double compared with the group-based instantiation. 
Furthermore, unlike the original scheme—where the client performs only lightweight local operations using plaintext $\mathbf{x}$ and the downloaded $\C_{\mathcal{M}}$—the client must now encrypt $\mathbf{x}$ under the more expensive RLWE-based \textsf{HSS} for every inference. 
Consequently, this variant is particularly well-suited for deployments characterized by: (i) a large and diverse client population where each client issues only a small number of queries, making amortization of the one-time downloading cost ineffective; (ii) resource-constrained or stateless clients (e.g., IoT devices, browser-based applications) that cannot maintain persistent local storage.

\end{document}